\documentclass[runningheads]{llncs}
\usepackage{amssymb}
\usepackage{amsmath}

\usepackage{thmtools}
\usepackage{thm-restate}

\usepackage{mathpartir}

\usepackage{booktabs}

\usepackage{ifpdf}

\usepackage{graphicx}

\usepackage{enumitem}

\usepackage{mathtools}

\usepackage{epstopdf}

\usepackage[labelformat=simple]{subcaption} %
\usepackage{listings}

\usepackage{scalerel}

\usepackage[dvipsnames]{xcolor}

\usepackage{rotating}

\usepackage{pifont}
\usepackage[T1]{fontenc}

\usepackage{tikz}
\usepackage{mdframed}
\usepackage{array}
\usepackage{multicol}
\usepackage{mwe}
\usepackage{wrapfig}
\usepackage{backnaur}
\usepackage{stmaryrd}
\usetikzlibrary{shapes,arrows,automata,positioning,cd,calc}
\usepgflibrary{plotmarks}
\tikzset{
  cfgedge/.style   = {black, ->, >=stealth},
  cfgedgestar/.style   = {black, ->, >=stealth, shorten >=.17em,
                          to path={-- node[inner sep=0pt,at end,sloped] {${}^*$} (\tikztotarget) \tikztonodes}},
  cfgedgeplus/.style   = {black, ->, >=stealth, shorten >=.17em,
                          to path={-- node[inner sep=0pt,at end,sloped] {${}^+$} (\tikztotarget) \tikztonodes}},
  forward/.style = { blue, ->, >=angle 45},
  forwardstar/.style = {blue, ->, >=angle 45, shorten >=.17em,
             to path={-- node[inner sep=0pt,at end,sloped] {${}^*$} (\tikztotarget) \tikztonodes}},
  forwardplus/.style = {blue, ->, >=angle 45, shorten >=.17em,
             to path={-- node[inner sep=0pt,at end,sloped] {${}^+$} (\tikztotarget) \tikztonodes}},
  backward/.style = { red, densely dashed, ->, >=latex' },
  backwardstar/.style = { red, densely dashed, ->, >=latex', shorten >=.17em,
                      to path={-- node[inner sep=0pt,at end,sloped] {${}^*$} (\tikztotarget) \tikztonodes}},
}

\newcommand{\cfgarrow}{\mathbin{\tikz[baseline]\path (0,0) edge[cfgedge, yshift=0.6ex] (.9em,0);}}

\usetikzlibrary{fit,calc}
\newcommand*{\tikzmk}[1]{\tikz[remember picture,overlay,] \node (#1) {};\ignorespaces}
\newcommand{\boxit}[2]{\tikz[remember picture,overlay]{\node[yshift=3pt,xshift=3pt,fill=#1,opacity=.15,fit={(A)($(B)+(#2\linewidth,.8\baselineskip)$)}] {};}\ignorespaces}

\colorlet{pink}{red!40}
\colorlet{cyan}{cyan!60}
\colorlet{gray}{gray!60}

\usepackage{prettyref}
\newcommand{\pref}{\prettyref}
\newrefformat{thm}{Theorem~\ref{#1}}
\newrefformat{lem}{Lemma~\ref{#1}}
\newrefformat{cor}{Corollary~\ref{#1}}
\newrefformat{cha}{Chapter~\ref{#1}}
\newrefformat{sec}{\S\!~\ref{#1}}
\newrefformat{app}{Appendix~\ref{#1}}
\newrefformat{tab}{Table~\ref{#1}}
\newrefformat{fig}{Figure~\ref{#1}}
\newrefformat{alg}{Algorithm~\ref{#1}}
\newrefformat{exa}{Example~\ref{#1}}
\newrefformat{def}{Definition~\ref{#1}}
\newrefformat{li}{Line~\ref{#1}}
\newrefformat{eq}{Eq.~\ref{#1}}
\newrefformat{case}{Case~\ref{#1}}

\newlist{ecomponents}{enumerate}{1}
\setlist[ecomponents,1]{label={\bfseries C\arabic*},align=left}

\newlist{casesp}{enumerate}{3} %
\setlist[casesp]{align=left, %
                 listparindent=\parindent, %
                 parsep=\parskip, %
                 font=\normalfont\bfseries, %
                 leftmargin=0pt, %
                 labelwidth=0pt, %
                 itemindent=.4em,labelsep=.4em, %
                 partopsep=0pt, %
                 }
\setlist[casesp,1]{label=Case~\arabic*,ref=\arabic*,leftmargin=0ex}
\setlist[casesp,2]{label=Case~\thecasespi.\alph*,ref=\thecasespi.\alph*,leftmargin=2ex}
\setlist[casesp,3]{label=Case~\thecasespii.\roman*,ref=\thecasespii.\roman*,leftmargin=2ex}

\usepackage{bm}

\usepackage[ruled,vlined,linesnumbered]{algorithm2e}

\definecolor{lapislazuli}{rgb}{0.15, 0.38, 0.61}
\SetAlgoNlRelativeSize{-1}

\SetCommentSty{CommentFont}
\SetKwComment{Comment}{$\triangleright$ }{}

\definecolor{light-gray}{gray}{0.9}
\definecolor{light-pink}{rgb}{0.858, 0.188, 0.478}

\definecolor{maroon}{rgb}{0.5, 0.0, 0.0}

\newcommand{\etch}{{\color{maroon}h}}
\newcommand{\cbox}[2]{\fboxsep 1pt%
    \colorbox{#1}{#2}}
\newcommand{\highlightmath}[1]{\cbox{light-gray}{$#1$}}

\makeatletter
\def\mathcolor#1#{\@mathcolor{#1}}
\def\@mathcolor#1#2#3{%
  \protect\leavevmode
  \begingroup
    \color#1{#2}#3%
  \endgroup
}
\makeatother
\newcommand*{\redv}{\mathcolor{red}{v}}

\newcommand{\Omit}[1]{}

\newcommand{\AbsDomain}{\mathcal{A}}

\newcommand{\xqed}[1]{%
 \leavevmode\unskip\penalty9999 \hbox{}\nobreak\hfill
  \quad\hbox{\ensuremath{#1}}}

\newcommand*{\qef}{\xqed{\scriptstyle{\blacksquare}}}

\newcommand{\subsubsubsection}[1]{\smallskip\noindent\textbf{\emph{#1}}\enspace}
\newcommand{\RQ}[1]{\textbf{RQ#1}}

\usepackage[xcolor]{changebar}
\cbcolor{red}

\newcommand{\nestingrelation}{\mathsf{N}}
\newcommand{\nesting}{\mathcolor{maroon}{\preceq_{\nestingrelation}}}
\newcommand{\nestingneq}{\mathcolor{maroon}{\prec_{\mathsf{N}}}}
\newcommand{\luo}{\mathcolor{PineGreen}{\leq}}
\newcommand{\luoneq}{\mathcolor{PineGreen}{<}}
\newcommand{\relation}{\mathsf{R}}
\newcommand{\widen}{\triangledown}

\newcommand{\wto}{\mathcal{W}}

\newcommand{\hto}{\mathcal{H}}

\newcommand{\components}{\mathcal{C}}

\def\llceil{\lceil\kern-3pt\lceil}
\def\rrceil{\rceil\kern-3pt\rceil}
\def\llfloor{\lfloor\kern-3pt\lfloor}
\def\rrfloor{\rfloor\kern-3pt\rfloor}

\newcommand{\postset}[1]{\llfloor {#1}\rceil}
\newcommand{\preset}[1]{\lfloor{#1}\rrceil}
\newcommand{\component}[2]{\llfloor {#1}, {#2}\rrceil}

\DeclareMathOperator*{\foo}{\texttt{Lift}}

\newcommand{\eqdef}{\buildrel \mbox{\tiny\textrm{def}} \over =}

\newcommand*{\mikos}{\textsc{Mikos}}
\newcommand*{\ikos}{\text{IKOS}}

\newcommand*{\pre}{\textsc{Pre}}
\newcommand*{\post}{\textsc{Post}}
\newcommand*{\fpcheck}{\varphi}
\newcommand*{\checkmap}{\textsc{Ck}}

\newcommand*{\fmProg}{\mathit{Prog}}
\newcommand*{\fmexec}{\mathtt{exec}}
\newcommand*{\fmrepeat}{\mathtt{repeat}}
\newcommand*{\fmseq}{\fatsemi} %
\newcommand*{\FM}{\texttt{FM}}
\newcommand*{\memconfig}{\mathcolor{blue}{\mathcal{M}}}
\newcommand*{\memconfigdef}{\memconfig_{\dflt}}
\newcommand*{\memconfigopt}{\memconfig_{\opt}}

\newcommand*{\checkvertices}{V_C}
\newcommand*{\genprog}{\mathtt{genProg}}

\SetKw{Delete}{free}
\SetKw{foreachKw}{foreach}
\SetKw{doKw}{do}
\SetKw{whileKw}{while}
\newcommand{\deallocpre}[1]{\Delete\ \pre[#1]}
\newcommand{\deallocpost}[1]{\Delete\ \post[#1]}

\newcommand*{\dpost}{\mathcolor{blue}{\textsc{Dpost}}}
\newcommand*{\achk}{\mathcolor{blue}{\textsc{Achk}}}
\newcommand*{\dpostl}{\mathcolor{blue}{\textsc{Dpost}^\ell}}
\newcommand*{\dprel}{\mathcolor{blue}{\textsc{Dpre}^\ell}}
\newcommand*{\instr}{\texttt{Inst}}

\newcommand*{\opt}{\mathcolor{red}{\text{opt}}}
\newcommand*{\dpostopt}{\dpost_{\opt}}
\newcommand*{\achkopt}{\achk_{\opt}}
\newcommand*{\dpostlopt}{\dpostl_{\opt}}
\newcommand*{\dprelopt}{\dprel_{\opt}}

\newcommand*{\dflt}{\text{dflt}}
\newcommand*{\dpostdef}{\dpost_{\dflt}}
\newcommand*{\achkdef}{\achk_{\dflt}}
\newcommand*{\dpostldef}{\dpostl_{\dflt}}
\newcommand*{\dpreldef}{\dprel_{\dflt}}

\DeclareMathSymbol{\qm}{\mathalpha}{operators}{"3F}
\DeclareMathSymbol{\lp}{\mathalpha}{operators}{"28}
\DeclareMathSymbol{\rp}{\mathalpha}{operators}{"29}
\DeclareMathSymbol{\cl}{\mathalpha}{operators}{"3A}
\DeclareMathAlphabet{\mathbbold}{U}{bbold}{m}{n}

\newcommand*\OR{\ |\ }

\newcommand{\algwto}{\texttt{Generate$\FM$Program}}

\newcommand{\timelimit}{an hour}%
\newcommand{\memlimit}{64 GB}%
\newcommand{\filtertime}{5 seconds}

\newcommand{\noncomplete}{\textcolor{red}{$\times$}}

\newcommand{\osstotal}{1503}
\newcommand{\ossmikosto}{76}
\newcommand{\ossmikosso}{8}
\newcommand{\ossikosfive}{994}
\newcommand{\ossfiltered}{426}

\newcommand{\ossfilteredok}{425}
\newcommand{\ossfilterednok}{1}

\newcommand{\ossmaxmem}{0.022}
\newcommand{\ossminmem}{0.998}
\newcommand{\ossgeomeanok}{0.436}
\newcommand{\ossgeomean}{0.437}

\newcommand{\ossgeomeanpercentage}{43.7}
\newcommand{\ossgeomeantimes}{2.29}

\newcommand{\osshistal}{44}
\newcommand{\osshistau}{184}
\newcommand{\osshistbl}{187}
\newcommand{\osshistbu}{436}
\newcommand{\osshistcl}{449}
\newcommand{\osshistcu}{1085}
\newcommand{\osshistdl}{1133}
\newcommand{\osshistdu}{64000}

\newcommand{\ossspeedupmin}{0.88$\times$}
\newcommand{\ossspeedupmax}{2.83$\times$}
\newcommand{\ossspeedupavg}{1.08$\times$}

\newcommand{\ossspeedupabsmin}{-409.74~\text{s}}
\newcommand{\ossspeedupabsmax}{198.39~\text{s}}
\newcommand{\ossspeedupabsavg}{1.29~\text{s}}

\newcommand{\ossspeedup}{331}

\newcommand{\ossshorttimemin}{0.06}
\newcommand{\ossshorttimemax}{4.98}
\newcommand{\ossshorttimeavg}{1.07}

\newcommand{\ossshortmemmin}{9}
\newcommand{\ossshortmemmax}{218}
\newcommand{\ossshortmemavg}{46}
\newcommand{\ossshorttimep}{+1.33}
\newcommand{\ossshorttimem}{-0.05}
\newcommand{\ossshorttimea}{+0.14}
\newcommand{\ossshortmemp}{+172}
\newcommand{\ossshortmemm}{-0.43}
\newcommand{\ossshortmema}{+18}

\newcommand{\svctotal}{2928}
\newcommand{\svcmikosto}{435}

\newcommand{\svcikosfive}{1709}
\newcommand{\svcfiltered}{784}

\newcommand{\svcfilteredok}{755}
\newcommand{\svcfilterednok}{29}

\newcommand{\svcmaxmem}{0.001}
\newcommand{\svcminmem}{0.895}
\newcommand{\svcgeomeanok}{0.044}
\newcommand{\svcgeomean}{0.041}
\newcommand{\svcgeomeanq}{0.009}

\newcommand{\svcgeomeanpercentage}{4.07}
\newcommand{\svcgeomeantimes}{24.57}

\newcommand{\svchistal}{52}
\newcommand{\svchistau}{642}
\newcommand{\svchistbl}{647}
\newcommand{\svchistbu}{2676}
\newcommand{\svchistcl}{2700}
\newcommand{\svchistcu}{9764}
\newcommand{\svchistdl}{9865}
\newcommand{\svchistdu}{64000}

\newcommand{\svcspeedupmin}{0.87$\times$}
\newcommand{\svcspeedupmax}{1.80$\times$}
\newcommand{\svcspeedupavg}{1.29$\times$}

\newcommand{\svcspeedupabsmin}{-7.47~\text{s}}
\newcommand{\svcspeedupabsmax}{1160.04~\text{s}}
\newcommand{\svcspeedupabsavg}{96.90~\text{s}}

\newcommand{\svcspeedup}{740}

\newcommand{\svcshorttimemin}{0.11}
\newcommand{\svcshorttimemax}{4.98}
\newcommand{\svcshorttimeavg}{0.58}

\newcommand{\svcshortmemmin}{25}
\newcommand{\svcshortmemmax}{564}
\newcommand{\svcshortmemavg}{42}
\newcommand{\svcshorttimep}{+1.44}
\newcommand{\svcshorttimem}{-0.61}
\newcommand{\svcshorttimea}{+0.08}
\newcommand{\svcshortmemp}{+490}
\newcommand{\svcshortmemm}{-0.37} %
\newcommand{\svcshortmema}{+12}

\newif\ifarxiv
\arxivtrue

\usepackage{graphicx}
\usepackage{hyperref}

\begin{document}
\title{Memory-Efficient Fixpoint Computation}
\author{Sung Kook Kim\inst{1} \and
Arnaud J.\ Venet\inst{2} \and
Aditya V.\ Thakur\inst{1}}
\authorrunning{S.\ Kim et al.}
\institute{University of California, Davis CA 95616, USA\\
\email{\{sklkim,avthakur\}@ucdavis.edu} \and
Facebook, Inc., Menlo Park CA 94025, USA\\
\email{ajv@fb.com}}
\maketitle              %
\begin{abstract}
  Practical adoption of static analysis often requires trading precision for
performance. This paper focuses on improving the memory efficiency of abstract
interpretation without sacrificing precision or time efficiency.
Computationally, abstract interpretation reduces the problem of inferring
program invariants to computing a fixpoint of a set of equations. This paper
presents a method to minimize the memory footprint in Bourdoncle's iteration
strategy, a widely-used technique for fixpoint computation.
Our technique is agnostic to the abstract domain used. We prove that our
technique is optimal (i.e., it results in minimum memory footprint) for
Bourdoncle's iteration strategy while computing the same result. We evaluate the
efficacy of our technique by implementing it in a tool called $\mikos$, which
extends the state-of-the-art abstract interpreter $\ikos$.
When verifying user-provided assertions, $\mikos$ shows a decrease in
peak-memory usage to $\svcgeomeanpercentage$\% ($\svcgeomeantimes\times$) on
average compared to $\ikos$. When performing interprocedural buffer-overflow
analysis, $\mikos$ shows a decrease in peak-memory usage to
$\ossgeomeanpercentage$\% ($\ossgeomeantimes\times$) on average compared to
$\ikos$.
\end{abstract}
\section{Introduction}
\label{sec:Introduction}

\emph{Abstract interpretation}~\cite{kn:CC77} is a general framework for
expressing static analysis of programs. Program invariants inferred by an
abstract interpreter are used in client applications such as program verifiers,
program optimizers, and bug finders. To extract the invariants, an abstract
interpreter computes a fixpoint of an equation system approximating the program
semantics. The efficiency and precision of the abstract interpreter depends on
the \emph{iteration strategy}, which specifies the order in which the equations
are applied during fixpoint computation.

The \emph{recursive iteration strategy} developed by
Bourdoncle~\cite{bourdoncle1993efficient} is widely used for fixpoint
computation in academic and industrial abstract interpreters such as NASA
IKOS~\cite{ikos2014}, Crab~\cite{crab}, Facebook SPARTA~\cite{sparta}, Kestrel
Technology CodeHawk~\cite{codehawk}, and Facebook
Infer~\cite{DBLP:conf/nfm/CalcagnoD11}. Extensions to Bourdoncle's approach that improve
precision~\cite{amatoSAS2013} and time
efficiency~\cite{DBLP:journals/pacmpl/KimVT20} have also been proposed.

This paper focuses on improving the memory efficiency of abstract interpretation.
This is an important problem in practice because large memory requirements can
prevent clients such as compilers and developer tools from using sophisticated
analyses.
This has motivated approaches for efficient implementations of abstract domains
\cite{DBLP:conf/cav/JeannetM09,DBLP:journals/scp/BagnaraHZ08,DBLP:conf/popl/SinghPV17},
including techniques that trade precision for efficiency
\cite{GangeNSSS16:VMCAI2016,DBLP:journals/sigsoft/BertraneCCFMMR11,DBLP:conf/sas/HeoOY16}.

This paper presents a technique for memory-efficient fixpoint computation. Our
technique minimizes the memory footprint in Bourdoncle's recursive iteration
strategy. Our approach is agnostic to the abstract domain and does not sacrifice
time efficiency. We prove that our technique exhibits optimal peak-memory usage
for the recursive iteration strategy while computing the same
fixpoint~(\pref{sec:Algorithm}). Specifically, our approach does not change the
iteration order but provides a mechanism for early deallocation of abstract
values. Thus, there is no loss of precision when improving memory performance.
Furthermore, such ``backward compatibility'' ensures that existing
implementations of Bourdoncle's approach can be replaced without impacting
clients of the abstract interpreter, an important requirement in practice.

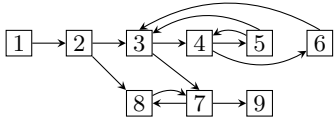
\begin{figure}[t]
  \centering
  \begin{tikzpicture}[auto,node distance=.8cm,font=\small]
    \tikzstyle{every node} = [rectangle, draw, inner sep=0pt,minimum size=2.5ex]
    \node (1) {$1$};
    \node [right of=1] (2)  {$2$};
    \node [right of=2] (3)  {$3$};
    \node [right of=3] (4)  {$4$};
    \node [right of=4] (5)  {$5$};
    \node [right of=5] (6)  {$6$};

    \node [below of=3] (8)  {$8$};
    \node [right of=8] (7)  {$7$};
    \node [right of=7] (9)  {$9$};

    \path (1) edge[cfgedge] (2);
    \path (2) edge[cfgedge] (3);
    \path (3) edge[cfgedge] (4);
    \path (4) edge[cfgedge] (5);
    \path (4) edge[cfgedge, bend right] (6);
    \path (6.north) edge[cfgedge, bend right] (3.north);
    \path (5) edge[cfgedge, bend right] (4);
    \path (5.north) edge[cfgedge, bend right] (3);%

    \path (3) edge[cfgedge] (7.north);
    \path (2) edge[cfgedge] (8);

    \path (7) edge[cfgedge] (8);
    \path (8) edge[cfgedge, bend left] (7);
    \path (7) edge[cfgedge] (9);
  \end{tikzpicture}
  \vspace{-1ex}
  \caption{Control-flow graph $G_1$}
\label{fig:ex}
\label{fig:cfg1}
\vspace{-5ex}
\end{figure}

Suppose we are tasked with proving assertions at program points $4$ and $9$ of
the control-flow graph $G_1(V, \cfgarrow)$ in \pref{fig:cfg1}. Current
approaches (\pref{sec:Bourdoncle}) allocate abstract values for each program
point during fixpoint computation, check the assertions at $4$ and $9$ after
fixpoint computation, and then deallocate all abstract values. In contrast, our
approach deallocates abstract values and checks the assertions during fixpoint
computation while guaranteeing that the results of the checks remain the same
and that the peak-memory usage is optimal.

We prove that our approach deallocates abstract values as soon as they are no
longer needed during fixpoint computation. Providing this theoretical guarantee
is challenging for arbitrary irreducible graphs such as $G_1$. For example,
assuming that node $8$ is analyzed after $3$, one might think that the fixpoint
iterator can deallocate the abstract value at $2$ once it analyzes $8$. However,
$8$ is part of the strongly-connected component $\{7, 8\}$, and the fixpoint
iterator might need to iterate over node $8$ multiple times. Thus, deallocating
the abstract value at $2$ when node $8$ is first analyzed will lead to incorrect
results. In this case, the earliest that the abstract value at $2$ can be
deallocated is after the stabilization of component $\{7, 8\}$.

Furthermore, we prove that our approach performs the assertion checks as early
as possible during fixpoint computation. Once the assertions are checked, the
associated abstract values are deallocated. For example, consider the assertion
check at node $4$. Notice that $4$ is part of the strongly-connected components
$\{4,5 \}$ and $\{3,4,5,6 \}$. Checking the assertion the first time node $4$ is
analyzed could lead to an incorrect result because the abstract value at $4$ has
not converged. The earliest that the check at node $4$ can be executed is after
the convergence of the component $\{3,4,5,6\}$. Apart from being able to
deallocate abstract values earlier, early assertion checks provide partial
results on timeout.

The key theoretical result (\pref{thm:full}) is that our iteration strategy is
memory-optimal (i.e., it results in minimum memory footprint) while computing
the same result as Bourdoncle's approach. Furthermore, we present an
almost-linear time algorithm to compute this optimal iteration
strategy~(\pref{sec:Efficient}).

We have implemented this memory-optimal fixpoint computation in a tool called
$\mikos$~(\pref{sec:Implementation}), which extends the state-of-the-art
abstract interpreter for C/C++, $\ikos$~\cite{ikos2014}. We compared the memory
efficiency of $\mikos$ and $\ikos$ on the following tasks:
\begin{itemize}
  \item[T1] Verifying user-provided assertions. Task~T1 represents the
  program-verification client of a fixpoint computation. We performed
  interprocedural analysis of \svcfiltered{} SV-COMP 2019
  benchmarks~\cite{svcomp} using reduced product of Difference Bound Matrix with
  variable packing~\cite{GangeNSSS16:VMCAI2016} and
  congruence~\cite{Granger:IJCM1989} domains.
  \item[T2] Proving absence of buffer overflows. Task~T2 represents the
  bug-finding and compiler-optimization client of fixpoint computation. In the
  context of bug finding, a potential buffer overflow can be reported to the
  user as a potential bug. In the context of compiler optimization, code to
  check buffer-access safety can be elided if the buffer access is verified to
  be safe. We performed interprocedural buffer overflow analysis of
  \ossfiltered{} open-source programs using the interval abstract domain.
\end{itemize}
\noindent On Task~T1, $\mikos$ shows a decrease in peak-memory usage to
$\svcgeomeanpercentage$\% ($\svcgeomeantimes\times$) on average compared to
$\ikos$. For instance, peak-memory required to analyze the SV-COMP 2019
benchmark \texttt{ldv-3.16-rc1/205\_9a-net-rtl8187} decreased from 46 GB to 56
\emph{MB}. Also, while \texttt{ldv-3.14/usb-mxl111sf} spaced out in $\ikos$ with
64 GB memory limit, peak-memory usage was 21 GB for $\mikos$. On Task~T2,
$\mikos$ shows a decrease in peak-memory usage to $\ossgeomeanpercentage$\%
($\ossgeomeantimes\times$) on average compared to $\ikos$. For instance,
peak-memory required to analyze a benchmark \texttt{ssh-keygen} decreased from
30 GB to 1 GB.

The contributions of the paper are as follows:
\begin{itemize}
  \item A memory-optimal technique for Bourdoncle's recursive iteration strategy
  that does not sacrifice precision or time efficiency~(\pref{sec:Algorithm}).
  \item An almost-linear time algorithm to construct our memory-efficient
  iteration strategy~(\pref{sec:Efficient}).
  \item $\mikos$, an interprocedural implementation of our
  approach~(\pref{sec:Implementation}).
  \item An empirical evaluation of the efficacy of $\mikos$ using a large set of
  C benchmarks~(\pref{sec:evaluation}).
\end{itemize}

\noindent \pref{sec:Preliminaries} presents necessary background on fixpoint
computation, including Bourdoncle's approach; \pref{sec:Related} presents
related work; \pref{sec:Conclusion} concludes.

\section{Fixpoint Computation Preliminaries}
\label{sec:Preliminaries}

This section presents background on fixpoint computation that will allow us to
clearly state the problem addressed in this paper~(\pref{sec:ProblemStatement}).
This section is not meant to capture all possible approaches to implementing
abstract interpretation. However, it does capture the relevant high-level
structure of abstract-interpretation implementations such as
IKOS~\cite{ikos2014}.

Consider an equation system $\Phi$ whose dependency graph is $G(V, \cfgarrow)$.
The graph $G$ typically reflects the control-flow graph of the program, though
this is not always true. The aim is to find the fixpoint of the
equation system~$\Phi$:
\begin{align}
  \label{eq:FixpointEquation}
  \pre[v] &=  \bigsqcup\left\{\post[p] \mid p \cfgarrow v\right\}&v\in V \\
  \post[v] &= \tau_v(\pre[v])&v\in V \nonumber
\end{align}
The maps $\pre \colon V \to \AbsDomain$ and $\post \colon V \to \AbsDomain$
maintain the abstract values at the beginning and end of each program point,
where $\AbsDomain$ is an abstract domain. The abstract transformer $\tau_v
\colon \AbsDomain \to \AbsDomain$ overapproximates the semantics of program
point $v \in V$. After fixpoint computation, $\pre[v]$ is an invariant for $v
\in V$.

Client applications of the abstract interpreter typically query these fixpoint
values to perform assertion checks, program optimizations, or report bugs. Let
$\checkvertices \subseteq V$ be the set of program points where such checks
are performed, and let $\fpcheck_v \colon \AbsDomain \to bool$ represent the
corresponding functions that performs the check for each $v \in \checkvertices$.
To simplify presentation, we assume that the check function merely returns
\texttt{true} or \texttt{false}. Thus, after fixpoint computation, the client
application computes $\fpcheck_v(\pre[v])$ for each $v \in \checkvertices$.

The exact least solution of the system \pref{eq:FixpointEquation} can be
computed using Kleene iteration provided $\AbsDomain$ is Noetherian. However,
most interesting abstract domains require the use of \emph{widening} ($\widen$)
to ensure termination followed by \emph{narrowing} to improve the post solution.
In this paper, we use ``fixpoint'' to refer to such an approximation of the
least fixpoint. Furthermore, for simplicity of presentation, we restrict our
description to a simple widening strategy. However, our implementation
(\pref{sec:Implementation}) uses more sophisticated widening and narrowing
strategies implemented in state-of-the-art abstract interpreters
\cite{ikos2014,amatoSAS2013}.

An \emph{iteration strategy} specifies the order in which the individual
equations are applied, where widening is used, and how convergence of the
equation system is checked. For clarity of exposition, we introduce a
\emph{Fixpoint Machine (\FM)} consisting of an imperative set of instructions.
An $\FM$ program represents a particular iteration strategy used for fixpoint
computation.
The syntax of Fixpoint Machine programs is defined by the following
grammar:
\begin{align}
  \label{eq:FixpointMachineSyntax}
  \fmProg ::= \fmexec\ v
                \OR  \fmrepeat\ v\ \texttt{[}\fmProg\texttt{]}
                \OR \fmProg \fmseq \fmProg \ \ \ \ , v \in V
\end{align}
Informally, the instruction $\fmexec\ v$ applies $\tau_v$ for $v \in V$; the
instruction $\fmrepeat\ v\ \texttt{[}P_1 \texttt{]} $ repeatedly executes the
$\FM$ program $P_1$ until convergence and performs widening at $v$; and the
instruction $P_1 \fmseq P_2$ executes $\FM$ programs $P_1$ and $P_2$ in
sequence.

The syntax (\pref{eq:FixpointMachineSyntax}) and semantics
(\pref{fig:FixpointMachineSemantics}) of the Fixpoint Machine are sufficient to
express Bourdoncle's recursive iteration strategy~(\pref{sec:Bourdoncle}), a
widely-used approach for fixpoint computation~\cite{bourdoncle1993efficient}. We
also extend the notion of iteration strategy to perform memory management of the
abstract values as well as perform checks during fixpoint
computation~(\pref{sec:MemoryManagement}).

\subsection{Bourdoncle's Recursive Iteration Strategy}
\label{sec:Bourdoncle}

In this section, we review Bourdoncle's recursive iteration
strategy~\cite{bourdoncle1993efficient} and show how to generate the
corresponding $\FM$ program.

Bourdoncle's iteration strategy relies on the notion of \emph{weak topological
ordering~(WTO)} of a directed graph $G(V, \cfgarrow)$. A WTO is defined using
the notion of a \emph{hierarchical total ordering (HTO)} of a set.

\begin{definition}
  \label{def:BourdoncleHTO}
  A \emph{hierarchical total ordering} $\hto$ of a set $S$ is a well
  parenthesized permutation of $S$ without two consecutive ``(''. \qef
\end{definition}
An HTO $\hto$ is a string over the alphabet $S$ augmented with left and right
parenthesis. Alternatively, we can denote an HTO $\hto$ by the tuple $(S,
\preceq, \omega)$, where $\preceq$ is the total order induced by $\hto$ over the
elements of $S$ and $\omega \colon V \to 2^V$. The elements between two matching
parentheses are called a \emph{component}, and the first element of a component
is called the \emph{head}. Given $l \in S$, $\omega(l)$ is the set of heads of
the components containing $l$. We use $\components \colon V \to 2^V$ to denote the
mapping from a head to its component.

\begin{example}
  \label{exa:HTO}
  Let $V = \{1, 2, 3, 4, 5, 6, 7, 8, 9\}$. An example HTO $\hto_1(V, \preceq,
  \omega)$ is\\
  $\mathtt{1\ 2}$ $\mathtt{(3\ (4\ 5)\ 6)\ (7\ 8)\ 9}$. $\omega(3) = \{3\}$,
  $\omega(5) = \{3, 4\}$, and $\omega(1) = \emptyset$. It has components
  $\components(4) = \{4, 5\}$, $\components(7) = \{7, 8\}$ and $\components(3) =
  \{3, 6\} \cup \components(4)$.
  \qef
\end{example}

A weak topological ordering (WTO) $\wto$ of a directed graph $G(V, \cfgarrow)$
is an HTO $\hto(V, \preceq, \omega)$ satisfying certain constraints listed
below:

\begin{definition}
  \label{def:BourdoncleWTO}
  A \emph{weak topological ordering} $\wto(V, \preceq, \omega)$ of a
  directed graph $G(V, \cfgarrow)$ is an HTO $\hto(V, \preceq, \omega)$ such
  that for every edge $u \rightarrow v$, either (i)~$u \prec v$, or (ii)~$v \preceq u$
  and $v \in \omega(u)$. \qef
\end{definition}

\begin{example}
  \label{exa:WTO}
  HTO $\hto_1$ in \pref{exa:HTO} is a WTO $\wto_1$ of the graph $G_1$
  (\pref{fig:cfg1}).
  \qef
\end{example}

Given a directed graph $G(V, \cfgarrow)$ that represents the dependency graph of
the equation system, Bourdoncle's approach uses a WTO $\wto(V, \preceq, \omega)$
of $G$ to derive the following \emph{recursive iteration strategy}:
\begin{itemize}
  \item The total order $\preceq$ determines the order in which the equations
  are applied. The equation after a component is applied only after the component
  stabilizes.
  \item The stabilization of a component $\components(h)$ is determined by
  checking the stabilization of the head $h$.
  \item Widening is performed at each of the heads.
\end{itemize}
We now show how the WTO can be represented using the syntax of our Fixpoint
Machine ($\FM$) defined in \pref{eq:FixpointMachineSyntax}.
The following function $\genprog \colon \text{WTO} \to \fmProg$ maps a given WTO
$\wto$ to an $\FM$ program:
\vspace{-1.2ex}
\begin{equation}
  \label{eq:genprog}
  \genprog(\wto) :=
  \begin{cases}
   \fmrepeat\ v\ \texttt{[} \genprog(\wto') \texttt{]} &\text{if } \wto = (v\ \wto') \\
   \genprog(\wto_1) \fmseq \genprog(\wto_2) &\text{if } \wto = \wto_1\ \wto_2 \\
   \fmexec\ v &\text{if } \wto = v
  \end{cases}
\end{equation}
Each node $v \in V$ is mapped to a single $\FM$ instruction by $\genprog$; we
use $\instr[v]$ to refer to this $\FM$ instruction corresponding to $v$. Note
that if $v\in V$ is a head, then $\instr[v]$ is an instruction of the form
$\fmrepeat\ v\ \texttt{[}\ldots \texttt{]}$, else $\instr[v]$ is $\fmexec\ v$.

\begin{example}
  \label{exa:map}
  The WTO $\wto_1$ of graph $G_1$ (\pref{fig:cfg1}) is $\mathtt{1\ 2\
  \mathcolor{red}{(}3\ \mathcolor{blue}{(}4\ 5\mathcolor{blue}{)}\
  6\mathcolor{red}{)}\ \mathcolor{ForestGreen}{(}7\ 8\mathcolor{ForestGreen}{)}\
  9}$.
  The corresponding $\FM$ program is $P_1 = \genprog(\wto_1) = \fmexec\ 1 \fmseq
  \fmexec\ 2 \fmseq\ \fmrepeat\ 3\ \mathcolor{red}{\texttt{[}}\fmrepeat\ 4\
  \mathcolor{blue}{\texttt{[}}\fmexec\ 5\mathcolor{blue}{\texttt{]}} \fmseq
  \fmexec\ 6\mathcolor{red}{\texttt{]}} \fmseq \fmrepeat\ 7\
  \mathcolor{ForestGreen}{\texttt{[}}\fmexec\
  8\mathcolor{ForestGreen}{\texttt{]}} \fmseq \fmexec\ 9$.
  The colors used for brackets and parentheses are to more clearly indicate the
  correspondence between the WTO and the $\FM$ program.
  Note that $\instr[1] = \fmexec\ 1$, and $\instr[4] = \fmrepeat\ 4
  \mathcolor{blue}{\texttt{ [}}\fmexec\ 5\mathcolor{blue}{\texttt{]}}$.
  \qef
\end{example}
Ignoring the text in \highlightmath{\text{gray}}, the semantics of the $\FM$
instructions shown in \pref{fig:FixpointMachineSemantics} capture Bourdoncle's
recursive iteration strategy. The semantics are parameterized by the graph $G(V,
\cfgarrow)$ and a WTO $\wto(V, \preceq, \omega)$.

\begin{figure}[t]
  \begin{equation*}
  \begin{array}{l l}
    G(V, \cfgarrow),& \text{WTO } \wto(V, \preceq, \omega),\\
    \highlightmath{V_C \subseteq V},& \highlightmath{\text{ memory configuration }\memconfig (\dpost, \achk, \dpostl, \dprel)}
  \end{array}
  \vspace{-2ex}
\end{equation*}
  \begin{align*}
    \llbracket \texttt{exec}\ \redv \rrbracket_{\highlightmath{\memconfig}} {\color{maroon}\eqdef}\ &
    \begin{aligned}[t]
        & \pre[\redv] \gets \bigsqcup\{\post[p] \mid p \cfgarrow \redv\} \\
        & \highlightmath{  \foreachKw\ u \in V \colon \redv = \dpost[u] \Rightarrow \deallocpost{u}  } \\
        & \post[\redv] \gets \tau_{\redv}(\pre[\redv]) \\
        &\highlightmath{\begin{aligned}
        & \redv \notin \checkvertices \Rightarrow \deallocpre{\redv} \\
        & \foreachKw\ u \in \checkvertices \colon \redv = \achk[u]\Rightarrow
                                \begin{aligned}[t]&\checkmap[u] \gets \fpcheck_{u}(\pre[u]); \\ &\deallocpre{u} \end{aligned}
        \end{aligned}}%
    \end{aligned}
    \\
    \llbracket \texttt{repeat}\ \redv\ \texttt{[}P \texttt{]} \rrbracket_{\highlightmath{\memconfig}} {\color{maroon}\eqdef} &
    \left.
    \begin{aligned}
         & tpre \gets \bigsqcup\{\post[p] \mid p \cfgarrow \redv \wedge \redv \notin \omega(p)\}
    \end{aligned}
    \vphantom{\frac{\frac{1}{2}}{3}}\bm{\right\}}\text{\scriptsize Preamble}
    \\ &
    \left.
    \begin{aligned}
            & \doKw \textbf{ \{} \\
            &\highlightmath{\begin{aligned}
            & \quad   \foreachKw\ u \in V \colon \redv \in \dpostl[u] \Rightarrow \deallocpost{u} \\
            & \quad   \foreachKw\ u \in \checkvertices \colon \redv \in \dprel[u] \Rightarrow \deallocpre{u}
            \end{aligned}}%
            \\
            & \quad \pre[\redv], \post[\redv] \gets tpre, \tau_{\redv}(tpre) \\
            & \quad \llbracket P \rrbracket_{\highlightmath{\memconfig}} \\
            & \quad tpre \gets \pre[\redv] \widen \bigsqcup\{\post[p] \mid p \cfgarrow v\} \\
            & \textbf{\} } \whileKw (tpre \not\sqsubseteq \pre[\redv])
    \end{aligned}
    \right\}\text{\scriptsize Loop}
    \\ &
    \left.
    \highlightmath{
    \begin{aligned}
        & \foreachKw\ u \in V \colon \redv = \dpost[u] \Rightarrow \deallocpost{u} \\
        & \redv \notin \checkvertices \Rightarrow \deallocpre{\redv} \\
        & \foreachKw\ u \in \checkvertices \colon \redv = \achk[u]\Rightarrow
                  \begin{aligned}[t]&\checkmap[u] \gets \fpcheck_{u}(\pre[u]); \\ &\deallocpre{u} \end{aligned}
    \end{aligned}
    }%
    \right\}\text{\scriptsize Postamble}
    \\
    \llbracket P_1 \fmseq P_2 \rrbracket_{\highlightmath{\memconfig}} {\color{maroon}\eqdef} &
      \begin{aligned}[t]
        & \llbracket P_1 \rrbracket_{\highlightmath{\memconfig}} \\
        & \llbracket P_2 \rrbracket_{\highlightmath{\memconfig}}%
      \end{aligned}
  \end{align*}
  \vspace{-3ex}
  \caption{The semantics of the Fixpoint Machine ($\FM$) instructions of \pref{eq:FixpointMachineSyntax}.}
  \label{fig:FixpointMachineSemantics}
\end{figure}
\vspace{-2ex}

\subsection{Memory Management during Fixpoint Computation}
\label{sec:MemoryManagement}

In this paper, we extend the notion of iteration strategy to indicate when
abstract values are deallocated and when checks are executed. The
\highlightmath{\text{gray}} text in \pref{fig:FixpointMachineSemantics} shows
the semantics of the $\FM$ instructions that handle these issues. The right-hand
side of $\Rightarrow$ is executed if the left-hand side evaluates to true.
Recall that the set $V_C \subseteq V$ is the set of program points that have
assertion checks. The map $\checkmap \colon V_C \to \texttt{bool}$ records the
result of executing the check $\fpcheck_u(\pre[u])$ for each $u \in V_C$. Thus,
the \emph{output of the $\FM$ program} is the map $\checkmap$. In practice, the
functions $\fpcheck_u$ are expensive to compute. Furthermore, they often write
the result to a database or report the output to a user. Consequently, we assume
that only the first execution of $\fpcheck_u$ is recorded in $\checkmap$.

The \emph{memory configuration} $\memconfig$ is a tuple $(\dpost, \achk,
\dpostl, \dprel)$ where
\begin{itemize}
  \item The map $\dpost \colon V \to V$ controls the deallocation of values in
  $\post$ that have no further use. If $v = \dpost[u]$, $\post[u]$ is
  deallocated after the execution of $\instr[v]$.
  \item The map $\achk \colon V_C \to V$ controls when the check function
  $\fpcheck_u$ corresponding to $u \in V_C$ is executed, after which the
  corresponding $\pre$ value is deallocated. If $\achk[u] = v$, assertions in
  $u$ are checked and $\pre[u]$ is subsequently deallocated after the execution
  of $\instr[v]$.
  \item The map $\dpostl \colon V \to 2^V$ control deallocation of $\post$ values
  that are recomputed and overwritten in the loop of a $\fmrepeat$ instruction
  before its next use. If $v \in \dpostl[u]$, $\post[u]$ is deallocated in the
  loop of $\instr[v]$.
  \item The map $\dprel \colon \checkvertices \to 2^V$ control deallocation of $\pre$
   values that recomputed and overwritten in the loop of a $\fmrepeat$
   instruction before its next use. If $v \in \dprel[u]$, $\pre[u]$ is
   deallocated in the loop of $\instr[v]$.
\end{itemize}

To simplify presentation, the semantics in \pref{fig:FixpointMachineSemantics}
does not make explicit the allocations of abstract values: if a $\post$ or
$\pre$ value that has been deallocated is accessed, then it is allocated and
initialized to $\bot$.

\subsection{Problem Statement}
\label{sec:ProblemStatement}
Two memory configurations are \emph{equivalent} if they result in the same
values for each check in the program:
\begin{definition}
  Given an $\FM$ program $P$, memory configuration $\memconfig_1$ is
  \emph{equivalent to} $\memconfig_2$, denoted by $\llbracket P
  \rrbracket_{\memconfig_1} = \llbracket P \rrbracket_{\memconfig_2}$, iff for
  all $u \in V_C$, we have $\checkmap_1[u] = \checkmap_2[u]$, where
  $\checkmap_1$ and $\checkmap_2$ are the check maps corresponding to execution
  of $P$ using $\memconfig_1$ and $\memconfig_2$, respectively.
  \qef
\end{definition}

The \emph{default memory configuration $\memconfigdef$} performs checks and
deallocations at the end of the $\FM$ program after fixpoint has been computed.
\begin{definition}
  Given an $\FM$ program $P$, the \emph{default memory configuration}
  $\memconfigdef$ $(\dpostdef, \achkdef, \dpostldef, \dpreldef)$ is
  $\dpostdef[v] = z$ for all $v \in V$, $\achkdef[c] = z$ for all $c \in
  \checkvertices$, and $\dpostldef = \dpreldef = \emptyset$, where $z$ is the
  last instruction in $P$.
  \qef
\end{definition}

\begin{example}
  Consider the $\FM$ program $P_1$ from \pref{exa:map}. Let $V_C = \{4, 9\}$.
  $\dpostdef[v] = 9$ for all $v \in V$. That is, all $\post$ values are
  deallocated at the end of the fixpoint computation. Also, $\achkdef[4] =
  \achkdef[9] = 9$, meaning that assertion checks also happen at the end.
  $\dpostldef = \dpreldef = \emptyset$, so the $\FM$ program does not clear
  abstract values whose values will be recomputed and overwritten in a loop of
  $\fmrepeat$ instruction. \qef
\end{example}

Given an $\FM$ program $P$, a memory configuration $\memconfig$ is \emph{valid}
for $P$ iff it is equivalent to the default configuration; i.e., $\llbracket P
\rrbracket_{\memconfig} = \llbracket P \rrbracket_{\memconfigdef}$.

Furthermore, a valid memory configuration $\memconfig$ is \emph{optimal} for a
given $\FM$ program iff memory footprint of $\llbracket P
\rrbracket_{\memconfig}$ is smaller than or equal to that of $\llbracket P
\rrbracket_{\memconfig'}$ for all valid memory configuration $\memconfig'$.
The problem addressed in this paper can be stated as:
\begin{mdframed}[backgroundcolor=light-gray] Given an $\FM$ program $P$, find an
  optimal memory configuration~$\memconfig$.
\end{mdframed}

An optimal configuration should deallocate abstract values during fixpoint
computation as soon they are no longer needed. The challenge is ensuring that
the memory configuration remains valid even without knowing the number of loop
iterations for $\fmrepeat$ instructions. \pref{sec:Algorithm} gives the optimal
memory configuration for the $\FM$ program $P_1$ from \pref{exa:map}.

\section{Declarative Specification of Optimal Memory Configuration $\memconfigopt$}
\label{sec:Algorithm}

This section provides a declarative specification of an optimal memory
configuration $\memconfigopt$($\dpostopt$, $\achkopt$, $\dpostlopt$,
$\dprelopt$). The proofs of the theorems in this section can be found in
\pref{app:appendix-proof}. \pref{sec:Efficient} presents an efficient algorithm
for computing $\memconfigopt$.

\begin{definition}
  \label{def:NestingRelation}
  Given a WTO $\wto(V, \preceq, \omega)$ of a graph $G(V, \cfgarrow)$, the
  \emph{nesting relation $\nestingrelation$} is a tuple $(V, \nesting)$ where $x
  \nesting y$ iff $x = y$ or $y \in \omega(x)$ for $x, y \in V$.\qef
\end{definition}
Let $\postset{v}_{\nesting} \eqdef \{ w\in V \mid v\ \nesting\ w\}$; that is,
$\postset{v}_{\nesting}$ equals the set containing $v$ and the heads of
components in the WTO that contain $v$.  The nesting relation
$\nestingrelation(V, \nesting)$ is a \emph{forest}; i.e. a partial order such that
for all $v\in V$,
$(\postset{v}_{\nesting}, \nesting)$ is a chain (\pref{thm:forest}, \pref{app:appendix-order}).

\begin{example}
  \label{exa:nesting}
  For the WTO $\wto_1$ of $G_1$ in \pref{exa:WTO}, $\nestingrelation_1(V,
  \nesting)$ is:
\raisebox{-20pt}{
  \begin{tikzpicture}[auto,node distance=.8cm,font=\small]
    \tikzstyle{every node} = [inner sep=0pt,minimum size=2ex, node distance=0.5cm]
    \node (1) {$1$};
    \node [right of=1, node distance=0.35cm] (2)  {$2$};
    \node [right of=2, node distance=0.35cm] (3)  {$3$};
    \node [right of=3, node distance=.8cm]   (7)  {$7$};
    \node [right of=7, node distance=0.35cm] (9)  {$9$};
    \node [below of=3] (4)  {$4$};
    \node [right of=4, node distance=0.35cm] (6)  {$6$};
    \node [below of=4] (5)  {$5$};
    \node [below of=7] (8)  {$8$};

    \path (3) edge (4);
    \path (4) edge (5);
    \path (3) edge (6);
    \path (7) edge (8);
  \end{tikzpicture}
}.
 Note that $\postset{5}_{\nesting} = \{5, 4, 3\}$, forming a chain $5 \nesting 4
  \nesting 3$.
  \qef
\end{example}

\subsection{Declarative Specification of $\dpostopt$}

$\dpostopt[u] = v$ implies that $v$ is the earliest instruction at
which $\post[u]$ can be deallocated while ensuring that there are no
subsequents reads of $\post[u]$ during fixpoint computation.
We cannot conclude $\dpostopt[u] = v$ from a dependency $u \cfgarrow v$ as
illustrated in the following example.

\begin{example}
\label{exa:nl}
Consider the $\FM$ program $P_1$ from \pref{exa:map}, whose graph $G_1(V,
\cfgarrow)$ is in \pref{fig:cfg1}. Although $2 \cfgarrow 8$, memory
configuration with $\dpost[2] = 8$ is not valid: $\post[2]$ is read by
$\instr[8]$, which is executed repeatedly as part of $\instr[7]$; if $\dpost[2]
= 8$, $\post[2]$ is deallocated the first time $\instr[8]$ is executed, and
subsequent executions of $\instr[8]$ will read $\bot$ as the value of
$\post[2]$.
\qef
\end{example}

In general, for a dependency $u \cfgarrow v$, we must find the head of maximal
component that contains $v$ but not $u$ as the candidate for $\dpostopt[u]$. By
choosing the head of \emph{maximal} component, we remove the possibility of
having a larger component whose head's $\fmrepeat$ instruction can execute
$\instr[v]$ after deallocating $\post[u]$. If there is no component that
contains $v$ but not $u$, we simply use $v$ as the candidate.
The following $\foo$ operator gives us the candidate of $\dpostopt[u]$ for
$u \cfgarrow v$:
\vspace{-1.2ex}
\begin{equation}
  \label{eq:Lift}
  \foo(u,v) \eqdef \text{max}_{\nesting}((\postset{v}_{\nesting}
  \setminus \postset{u}_{\nesting}) \cup \{v\})
\end{equation}
$\postset{v}_{\nesting}$ gives us $v$ and the heads of components that
contain~$v$. Subtracting $\postset{u}_{\nesting}$ removes the heads of
components that also contain~$u$. We put back $v$ to account for the case when
there is no component containing $v$ but not $u$ and $\postset{v}_{\nesting}
\setminus \postset{u}_{\nesting}$ is empty. Because $\nestingrelation(V,
\nesting)$ is a forest, $\postset{v}_{\nesting}$ and $\postset{u}_{\nesting}$
are chains, and hence, $\postset{v}_{\nesting} \setminus \postset{u}_{\nesting}$
is also a chain. Therefore, maximum is well-defined.

\begin{example}
  Consider the nesting relation $\nestingrelation_1(V, \nesting)$ from \pref{exa:nesting}.
  $\foo(2, 8)$\\$ = \text{max}_{\nesting}((\{8, 7\} \setminus \{2\}) \cup \{8\}) =
  7$. We see that $7$ is the head of the maximal component containing $8$ but
  not $2$.
  Also, $\foo(5, 4) = \text{max}_{\nesting}((\{4, 3\} \setminus \{5, 4, 3\})
  \cup \{4\}) = 4$. There is no component that contains $4$ but not $5$.
  \qef
\end{example}

For each instruction $u$, we now need to find the last instruction from among
the candidates computed using $\foo$. Notice that deallocations of $\post$
values are at a postamble of $\fmrepeat$ instructions in
\pref{fig:FixpointMachineSemantics}. Therefore, we cannot use the total order
$\preceq$ of a WTO to find the last instruction: $\preceq$ is the order in which
the instruction begin executing, or the order in which \emph{preamble}s are
executed.
\begin{example}
\label{exa:to}
Let $\dpost_{to}[u] \eqdef \text{max}_{\preceq}\{\foo(u, v) \mid u \cfgarrow
v\}, u \in V$, an incorrect variant of $\dpostopt$ that uses the total order
$\preceq$. Consider the $\FM$ program $P_1$ from \pref{exa:map}, whose graph
$G_1(V, \cfgarrow)$ is in \pref{fig:cfg1} and nesting relation
$\nestingrelation_1(V, \nesting)$ is in \pref{exa:nesting}. $\post[5]$ has
dependencies $5 \cfgarrow 4$ and $5 \cfgarrow 3$. $\foo(5, 4) = 4$, $\foo(5, 3)
= 3$. Now, $\dpost_{to}[5] = 4$ because $3 \preceq 4$. However, a memory
configuration with $\dpost[5] = 4$ is not valid: $\instr[4]$ is nested in
$\instr[3]$. Due to the deletion of $\post[5]$ in $\instr[4]$, $\instr[3]$ will
read $\bot$ as the value of $\post[5]$.
\qef
\end{example}

To find the order in which the instructions finish executing, or the order in
which \emph{postamble}s are executed, we define the relation $(V, \luo)$, using
the total order $(V, \preceq)$ and the nesting relation $(V, \nesting)$:
\vspace{-1.2ex}
\begin{equation}
  \label{eq:TotalOrder}
  x \luo y \eqdef x \nesting y \vee (y \not\nesting x \wedge x \preceq y)
  \vspace{-1.2ex}
\end{equation}
In the definition of $\luo$, the nesting relation $\nesting$ takes precedence
over $\preceq$. $(V, \luo)$ is a total order (\pref{thm:total},
\pref{app:appendix-order}). Intuitively, the total order $\luo$ moves the heads
in the WTO to their corresponding closing parentheses `)'.

\begin{example}
  \label{exa:luo}
  For $G_1$ (\pref{fig:cfg1}) and its WTO $\wto_1$, $\mathtt{1\ 2\ (3\ (4\ 5)\
  6)\ (7\ 8)\ 9}$, we have $1 \luo 2 \luo 5 \luo 4 \luo 6 \luo 3 \luo 8 \luo 7
  \luo 9$.
  Note that $3 \preceq 6$ while $6 \luo 3$. Postamble of $\fmrepeat\ 3\
  \texttt{[}\ldots\texttt{]}$ is executed after $\instr[6]$, while preamble of
  $\fmrepeat\ 3\ \texttt{[}\ldots\texttt{]}$ is executed before $\instr[6]$.
  \qef
\end{example}

We can now define $\dpostopt$. Given a nesting relation $\nestingrelation(V,
\nesting)$ for the graph $G(V, \cfgarrow)$, $\dpostopt$ is defined as:
\vspace{-1.2ex}
\begin{equation}
  \label{eq:dpostopt}
  \dpostopt[u] \eqdef \text{max}_{\luo}\{\foo(u, v) \mid u \cfgarrow v\}
  \ \ \ \ , u \in V
\end{equation}

\begin{example}
  \label{exa:dpostopt}
  Consider the $\FM$ program $P_1$ from \pref{exa:map}, whose graph $G_1(V,
  \cfgarrow)$ is in \pref{fig:cfg1} and nesting relation $\nestingrelation_1(V,
  \nesting)$ is in \pref{exa:nesting}.
  An optimal memory configuration $\memconfigopt$ defined by \pref{eq:dpostopt}
  is:
  \vspace{-1ex}
  \begin{equation*}
  \begin{aligned}
  \dpostopt[1] &= 2,\ \dpostopt[2] = \dpostopt[3] = \dpostopt[8] = 7,\ \dpostopt[4] = 6,\\
  \dpostopt[5] &=  \dpostopt[6] = 3,\ \dpostopt[7] = \dpostopt[9] = 9.
  \end{aligned}
  \end{equation*}

  Successors of $u$ are first lifted to compute $\dpostopt[u]$. For example, to
  compute $\dpostopt[2]$, $2$'s successors, $3$ and $8$, are lifted to $\foo(2,
  3) = 3$ and $\foo(2, 8) = 7$. To compute $\dpostopt[5]$, $5$'s successors, $3$
  and $4$, are lifted to $\foo(5, 3) = 3$ and $\foo(5, 4) = 4$. Then, the
  maximum (as per the total order $\luo$) of the lifted successors is chosen as
  $\dpostopt[u]$. Because $3 \luo 7$, $\dpostopt[2] = 7$. Thus, $\post[2]$ is
  deleted in $\instr[7]$. Also, because $4 \luo 3$, $\dpostopt[5] = 3$, and
  $\post[5]$ is deleted in $\instr[3]$.
  \qef
\end{example}

\subsection{Declarative Specification of $\achkopt$}

$\achkopt[u] = v$ implies that $v$ is the earliest instruction at which the
assertion check at $u \in \checkvertices$ can be executed so that the invariant
passed to the assertion check function $\fpcheck_u$ is the same as when using
$\memconfigdef$. Thus, guaranteeing the same check result $\checkmap$.

Because an instruction can be executed multiple times in a loop, we cannot
simply execute the assertion checks right after the instruction, as illustrated
by the following example.
\begin{example}
  \label{exa:achkoptwrong}
  Consider the $\FM$ program $P_1$ from \pref{exa:map}. Let $V_C = \{4,
  9\}$. A memory configuration with $\achk[4] = 4$ is not valid: $\instr[4]$ is
  executed repeatedly as part of $\instr[3]$, and the first value of $\pre[4]$
  may not be the final invariant.
  Consequently, executing $\fpcheck_4(\pre[4])$ in $\instr[4]$ may not give the
  same result as executing it in $\instr[9]$ ($\achkdef[4] = 9$).
  \qef
\end{example}

In general, because we cannot know the number of iterations of the loop in a
$\fmrepeat$ instruction, we must wait for the convergence of the maximal
component that contains the assertion check. After the maximal component
converges, the $\FM$ program never visits the component again, making $\pre$
values of the elements inside the component final. Only if the element is not in
any component can its assertion check be executed right after its instruction.

Given a nesting relation $\nestingrelation (V, \nesting)$ for the graph $G(V,
\cfgarrow)$, $\achkopt$ is defined as:
\vspace{-1.2ex}
\begin{equation}
  \label{eq:achkopt}
  \achkopt[u] \eqdef \text{max}_{\nesting}\postset{u}_{\nesting}
  \ \ \ \ , u \in \checkvertices
\end{equation}
Because $\nestingrelation (V,\nesting)$ is a forest, $(\postset{u}_{\nesting},
\nesting)$ is a chain. Hence, $\text{max}_{\nesting}$ is well-defined.

\begin{example}
  \label{exa:achkopt}
  Consider the $\FM$ program $P_1$ from \pref{exa:map}, whose graph $G_1(V,
  \cfgarrow)$ is in \pref{fig:cfg1} and nesting relation $\nestingrelation_1(V,
  \nesting)$ is in \pref{exa:nesting}. Suppose that $\checkvertices = \{4, 9\}$.
  $\achkopt[4] = \text{max}_{\nesting}\{4, 3\} = 3$ and $\achkopt[9] =
  \text{max}_{\nesting}\{9\} = 9$.
  \qef
\end{example}

\subsection{Declarative Specification of $\dpostlopt$}

$v \in \dpostl[u]$ implies that $\post[u]$ can be deallocated at $v$ because it
is recomputed and overwritten in the loop of a $\fmrepeat$ instruction before a
subsequent use of $\post[u]$. 

$\dpostlopt[u]$ must be a subset of $\postset{u}_{\nesting}$: only the
instructions of the heads of components that contain $v$ recompute $\post[u]$.
We can further rule out the instruction of the heads of components that contain
$\dpostopt[u]$, because $\instr[\dpostopt[u]]$ deletes $\post[u]$.
We add back $\dpostopt[u]$ to $\dpostlopt$ when $u$ is contained in
$\dpostopt[u]$, because deallocation by $\dpostopt$ happens after the
deallocation by $\dpostlopt$.

Given a nesting relation $\nestingrelation (V, \nesting)$ for the graph $G(V,
\cfgarrow)$, $\dpostlopt$ is defined as:
\vspace{-1.2ex}
\begin{equation}
  \label{eq:dpostlopt}
  \dpostlopt[u] \eqdef (\postset{u}_{\nesting} \setminus \postset{d}_{\nesting}) \cup \mathbbold{\lp} u \nesting d\ \mathbbold{\qm}\ \{d\}\ \mathbbold{\cl}\ \emptyset \mathbbold{\rp}
  \ \ \ \ , u \in V
\end{equation}
where $d = \dpostopt[u]$ as defined in \pref{eq:dpostopt}, and $\mathbbold{\lp}
\texttt{b}\ \mathbbold{\qm}\ \texttt{x}\ \mathbbold{\cl}\ \texttt{y}
\mathbbold{\rp}$ is the ternary conditional choice operator.

\begin{example}
  Consider the $\FM$ program $P_1$ from \pref{exa:map}, whose graph $G_1(V,
  \cfgarrow)$ is in \pref{fig:cfg1}, nesting relation $\nestingrelation_1(V,
  \nesting)$ is in \pref{exa:nesting}, and $\dpostopt$ is in
  \pref{exa:dpostopt}.
  \vspace{-1ex}
  \begin{equation*}
  \begin{aligned}
  \dpostlopt[1] &= \{1\},\ \dpostlopt[2] = \{2\},\ \dpostlopt[3] = \{3\},\\
  \dpostlopt[4] &= \{4\},\ \dpostlopt[5] = \{3, 4, 5\},\ \dpostlopt[6] = \{3, 6\},\\
  \dpostlopt[7] &= \{7\},\ \dpostlopt[8] = \{7, 8\},\ \dpostlopt[9] = \{9\}.
  \end{aligned}
  \end{equation*}

  For $7$, $\dpostopt[7] = 9$. Because $7 \not\nesting 9$, $\dpostlopt[7] =
  \postset{7}_{\nesting} \setminus \postset{9}_{\nesting} = \{7\}$. Therefore,
  $\post[7]$ is deleted in each iteration of the loop of $\instr[7]$.
  While $\instr[9]$ reads $\post[7]$ in the future, the particular values of
  $\post[7]$ that are deleted by $\dpostlopt[7]$ are not used in $\instr[9]$.
  For $5$, $\dpostopt[5] = 3$. Because $5 \nesting 3$, $\dpostlopt[5]$ =
  $\postset{5}_{\nesting} \setminus \postset{3}_{\nesting} \cup \{3\}$ = $\{5,
  4, 3\}$.
   \qef
\end{example}

\subsection{Declarative Specification of $\dprelopt$}

$v \in \dprel[u]$ implies that $\pre[u]$ can be deallocated at $v$ because it is
recomputed and overwritten in the loop of a $\fmrepeat$ instruction before a
subsequent use of $\pre[u]$. 

$\dprelopt[u]$ must be a subset of $\postset{u}_{\nesting}$: only the
instructions of the heads of components that contain $v$ recompute $\pre[u]$. If
$\instr[u]$ is a $\fmrepeat$ instruction, $\pre[u]$ is required to perform
widening. Therefore, $u$ must not be contained in $\dprelopt[u]$.
\begin{example}
  \label{exa:dpreloptwrong}
  Consider the $\FM$ program $P_1$ from \pref{exa:map}. Let $V_C = \{4,
  9\}$. A memory configuration with $\dprel[4] = \{3, 4\}$ is not valid, because
  $\instr[4]$ would read $\bot$ as the value of $\post[4]$ when performing
  widening.
  \qef
\end{example}

Given a nesting relation $\nestingrelation (V, \nesting)$ for the graph $G(V,
\cfgarrow)$, $\dprelopt$ is defined as:
\vspace{-1.2ex}
\begin{equation}
  \label{eq:dprelopt}
  \dprelopt[u] \eqdef \postset{u}_{\nesting} \setminus \{u\}
  \ \ \ \ , u \in \checkvertices
\end{equation}
\begin{example}
  Consider the $\FM$ program $P_1$ from \pref{exa:map}, whose graph $G_1(V,
  \cfgarrow)$ is in \pref{fig:cfg1} and nesting relation $\nestingrelation_1(V,
  \nesting)$ is in \pref{exa:nesting}. Let $\checkvertices = \{4, 9\}$.
  $\dprelopt[4] = \{4, 3\} \setminus \{4\} = \{3\}$ and $\dprelopt[9] = \{9\}
  \setminus \{9\} = \emptyset$. Therefore, $\pre[4]$ is deleted in each loop
  iteration of $\instr[3]$.
  \qef
\end{example}

The following theorem is proved in \pref{app:appendix-memopt}:
\begin{restatable}{theorem}{ThmMemopt}
  \label{thm:full}
  The memory configuration \emph{$\memconfigopt$($\dpostopt$, $\achkopt$,
  $\dpostlopt$, $\dprelopt$)} is optimal.
\end{restatable}
\section{Efficient Algorithm to Compute $\memconfigopt$}
\label{sec:Efficient}

\SetKwProg{Def}{def}{:}{end} \SetKwFunction{DKw}{D} \SetKwFunction{BKw}{B}
\SetKwFunction{CFKw}{CF} \SetKwFunction{nesttr}{P} \SetKwFunction{RKw}{R}
\SetKwFunction{TKw}{T} \SetKwFunction{repKw}{rep} \SetKwFunction{mergeKw}{merge}
\SetKw{inKw}{in} \SetKw{suchKw}{such} \SetKw{thatKw}{that}
\SetKw{thereKw}{there} \SetKw{stKw}{s.t.} \SetKw{existsKw}{exists}
\SetKw{andKw}{and} \SetKw{continueKw}{continue} \SetKw{returnKw}{return}
\SetKw{questionKw}{?} \SetKw{colonKw}{:}
\SetKwFunction{rmCFKw}{removeAllCrossFwdEdges}
\SetKwFunction{genFMKw}{generate$\FM$Instruction}
\SetKwFunction{rsCFKw}{restoreCrossFwdEdges}
\SetKwFunction{maxSCCKw}{connect$\FM$Instructions}
\SetKwFunction{FindNestedSCCsKw}{findNestedSCCs}
\SetKwFunction{DepthFirstForestKw}{DepthFirstForest}

\begin{figure}
  \begin{minipage}[t]{0.96\textwidth}
  \SetInd{.3em}{0.5em}
  \vspace{0pt}
  \begin{algorithm}[H]
    \small
    \DontPrintSemicolon
    \setlength{\columnsep}{15pt}
    \Indmm
    \KwIn{Directed graph $G(V, \cfgarrow)$}
    \KwOut{$\FM$ program $pgm$,
           $\memconfigopt(\dpostopt, \achkopt, \dpostlopt, \dprelopt)$}

    \Indpp
    \begin{multicols}{2}
    \small
    \tikzmk{A}
    $\DKw \coloneqq \DepthFirstForestKw(G)$\;
    $\cfgarrow_{\BKw} \coloneqq \textrm{back edges in $\DKw$}$\label{li:wto-B}\;
    $\cfgarrow_{\CFKw} \coloneqq \textrm{cross \& forward edges in $\DKw$}$\label{li:wto-CF}\;
    $\cfgarrow' \coloneqq \cfgarrow \setminus \cfgarrow_{\BKw}$\;
    \lFor{$v \in V$}{
      $\repKw(v) \coloneqq v$; $\RKw[v] \coloneqq \emptyset$\label{li:wto-init}
    }
    $\nesttr \coloneqq \emptyset$\;
    \rmCFKw{}\label{li:wto-rm}\;
    \For{$\etch \in V$ \inKw \text{descending} $\texttt{DFN}_{\DKw}$}{\label{li:wto-loop}
      \rsCFKw{$\etch$}\label{li:wto-rs}\;
      \genFMKw{$\etch$}\label{li:wto-scc}
    }
    $pgm \coloneqq \maxSCCKw{}$\label{li:wto-mscc}\;
    \returnKw{$pgm$, $\memconfigopt$}\;
    \BlankLine
    \tikzmk{B} \boxit{gray}{.94}
    \tikzmk{A}
    \Def{\rmCFKw{}}{
    \tikzmk{B} \boxit{cyan}{.89}
      \For{$(u, v) \in \cfgarrow_{\CFKw}$}{
        $\cfgarrow' \coloneqq \cfgarrow' \setminus \{ (u,v) \}$\label{li:wto-rmcf}\;
        \Comment*[l]{Lowest common ancestor.}
        $\RKw[\texttt{lca}_{\DKw}(u,v)] \coloneqq \RKw[\texttt{lca}_{\DKw}(u,v)] \cup \{(u,v)\}$\label{li:wto-lca}\;
      }
    }
    \tikzmk{A}
    \Def{\rsCFKw{$\etch$}}{
    \tikzmk{B} \boxit{cyan}{.89}
      $\cfgarrow' \coloneqq \cfgarrow' \cup \{(u, \repKw(v)) \mid (u, v) \in \RKw[\etch]\}$\label{li:wto-rscf}
    }
    \tikzmk{A}
    \Def{\FindNestedSCCsKw{$\etch$}}{
    \tikzmk{B} \boxit{cyan}{.89}
    $B_\etch \coloneqq \{ \repKw(p) \mid (p, \etch) \in \BKw\}$\label{li:wto-bp}\;
    $N_\etch \coloneqq \emptyset$\Comment*[r]{Nested SCCs except $\etch$.}%
    $W \coloneqq B_\etch \setminus \{ \etch \}$\Comment*[r]{Worklist.}
    \While{$\thereKw\ \existsKw\ v \in W$}{
      $W,\ N_\etch \coloneqq W \setminus \{v\},\ N_\etch \cup [v]$\label{li:wto-pop}\;
      \For{$u$ \stKw $u \cfgarrow' v$}{
        \If{$\repKw(u) \notin N_\etch \cup \{\etch \} \cup W$}{
          $W \coloneqq W \cup \{\repKw(u)\}$\label{li:wto-found}\;
        }
      }
    }
    \returnKw $N_\etch, B_\etch$
    }
    \tikzmk{A}
    \Def{\genFMKw{$\etch$}}{
      \tikzmk{B} \boxit{cyan}{.89}
      $N_\etch, B_\etch \coloneqq$ \FindNestedSCCsKw{$\etch$}\label{li:wto-find}\;
      \If{$B_\etch = \emptyset$}{
        $\instr[\etch] \coloneqq \fmexec\ \etch$\label{li:wto-exec}\;
        \returnKw \label{li:wto-ret}
      }
      \For{$v \in N_\etch$ \inKw \text{desc.} $\texttt{postDFN}_{\DKw}$}{\label{li:wto-lb}
        $\instr[\etch] \coloneqq \instr[\etch] \fmseq \instr[v]$\label{li:wto-repeatbody}\;
        {%
        \SetNlSty{textbf}{$\star$}{}
        \For{$u$ \stKw $u \cfgarrow' v$}{
          $\dpostopt[u] \coloneqq v$\label{li:wto-prec2}\;
          $\TKw[u] \coloneqq \repKw(u)$\label{li:wto-Tprec2}
        }
        }%
        \SetInd{.1em}{0.22em}
      }
      $\instr[\etch] \coloneqq \fmrepeat\ \etch \texttt{ [} \instr[\etch] \texttt{]}$\label{li:wto-repeat}\;
      {%
      \SetNlSty{textbf}{$\star$}{}
      \For{$u$ \stKw $u \cfgarrow_{\BKw} \etch$}{
        $\dpostopt[u] \coloneqq \TKw[u] \coloneqq \etch$\label{li:wto-nest}\label{li:wto-Tnest}
      }
      }%
      \For{$v \in N_\etch$}{
        \mergeKw{$v, \etch$}\label{li:wto-merge};
        $\nesttr \coloneqq \nesttr \cup \{(v, \etch)\}$\label{li:wto-nesttr}
      }
    }
    \tikzmk{A}
    \Def{\maxSCCKw{}}{
    \tikzmk{B} \boxit{cyan}{.89}
    $pgm := \epsilon$\Comment*[r]{Empty program.}
    \For{$v \in V$ \inKw desc. $\texttt{postDFN}_{\DKw}$}{\label{li:wto-lbb}
      \If{$\repKw(v) = v$} {
        $pgm := pgm \fmseq \instr[v]$\label{li:wto-final}\;
        {%
        \SetNlSty{textbf}{$\star$}{}
        \For{$u$ \stKw $u \cfgarrow' v$}{
          $\dpostopt[u] \coloneqq v$\label{li:wto-prec1}\;
          $\TKw[u] \coloneqq \repKw(u)$\label{li:wto-Tprec1}
        }
        }
      }
      {%
      \SetNlSty{textbf}{$\star$}{}
      \If{$v \in \checkvertices$}{
        $\achkopt[v] \coloneqq \repKw(v)$\label{li:wto-achk}\;
        $\dprelopt[v] \coloneqq \component{v}{\repKw(v)}_{\nesttr^*} \setminus \{v\}$\label{li:wto-dprel}
      }
      }%
      \SetInd{.1em}{0.22em}
    }
    {%
    \SetNlSty{textbf}{$\star$}{}
    \For{$v \in V$}{
      $\dpostlopt[v] \coloneqq \component{v}{T[v]}_{\nesttr^*}$\label{li:wto-dpostl}
    }
    }%
    \returnKw $pgm$
    }%
    \end{multicols}
    \vspace{-2ex}
    \caption{\algwto{}$(G)$}
    \label{alg:bu}
  \end{algorithm}
  \SetInd{.5em}{1em}
  \end{minipage}
  \vspace{-1ex}
\end{figure}

Algorithm \algwto{} (\pref{alg:bu}) is an almost-linear time algorithm for
computing an $\FM$ program $P$ and optimal memory configuration $\memconfigopt$
for a given directed graph $G(V, \cfgarrow)$. \pref{alg:bu} adapts the bottom-up
WTO construction algorithm presented in
Kim~et~al.~\cite{DBLP:journals/pacmpl/KimVT20}. In particular, \pref{alg:bu}
applies the $\genprog$ rules (\pref{eq:genprog}) to generate the $\FM$ program
from a WTO.
\pref{li:wto-exec} generates $\fmexec$ instructions for non-heads.
\pref{li:wto-repeat} generates $\fmrepeat$ instructions for heads, with their
bodies ([ ]) generated on \pref{li:wto-repeatbody}. Finally, instructions are
merged on \pref{li:wto-final} to construct the final output $P$.

Algorithm \algwto{} utilizes a disjoint-set data structure. Operation
$\repKw(v)$ returns the representative of the set that contains $v$. In
\pref{li:wto-init}, the sets are initialized to be $\repKw(v) = v$ for all $v
\in V$. Operation $\mergeKw(v,\etch)$ on \pref{li:wto-merge} merges the sets
containing $v$ and $\etch$, and assigns $\etch$ to be the representative for the
combined set.
$\texttt{lca}_{\DKw}(u,v)$ is the lowest common ancestor of $u, v$ in the
depth-first forest $D$~\cite{DBLP:journals/jacm/Tarjan79}. Cross and forward
edges are initially removed from $\cfgarrow'$ on \pref{li:wto-rm}, making the
graph $(V, \cfgarrow' \cup \cfgarrow_{\BKw})$ reducible. Restoring it on
\pref{li:wto-rs} when $\etch = \texttt{lca}_{\DKw}(u,v)$ restores some
reachability while keeping $(V, \cfgarrow' \cup \cfgarrow_{\BKw})$ reducible.

Lines indicated by $\star$ in \pref{alg:bu} compute $\memconfigopt$. Lines
\ref{li:wto-prec2}, \ref{li:wto-nest}, and \ref{li:wto-prec1} compute
$\dpostopt$. Due to the specific order in which the algorithm traverses $G$,
$\dpostopt[u]$ is overwritten with greater values (as per the total order
$\luo$) on these lines, making the final value to be the maximum among the
successors. $\foo$ is implicitly applied when restoring the edges in \rsCFKw:
edge $u \cfgarrow v$ whose $\foo(u, v) = \etch$ is replaced to $u \cfgarrow'
\etch$ on \pref{li:wto-rs}.

$\dpostlopt$ is computed using an auxiliary map $\TKw \colon V \to V$ and a relation
$\nesttr \colon V \times V$. At the end of the algorithm, $\TKw[u]$ will be the
maximum element (as per $\nesting$) in $\dpostlopt[u]$. That is, $\TKw[u] =
\text{max}_{\nesting}((\postset{u}_{\nesting} \setminus \postset{d}_{\nesting})
\cup \mathbbold{\lp} u \nesting d\ \mathbbold{\qm}\ \{d\}\ \mathbbold{\cl}\
\emptyset \mathbbold{\rp})$, where $d = \dpostopt[u]$. Once $\TKw[u]$ is
computed by lines \ref{li:wto-Tprec2}, \ref{li:wto-Tnest}, and
\ref{li:wto-Tprec1}, the transitive reduction of $\nesting$, $\nesttr$, is used
to find all elements of $\dpostlopt[u]$ on \pref{li:wto-dpostl}. $\nesttr$ is
computed on \pref{li:wto-nesttr}. Note that $\nesttr^* = \nesting$ and
$\component{x}{y}_{\nesttr^*} \eqdef \{v \mid x\ \nesttr^* v \wedge v\ \nesttr^*
y \}$. $\achk$ and $\dprel$ are computed on Lines~\ref{li:wto-achk} and
\ref{li:wto-dprel}, respectively.

\begin{example}
\label{exa:effex}
Consider the graph $G_1$ (\pref{fig:cfg1}).
Labels of vertices indicate a depth-first numbering (\texttt{DFN}) of $G_1$.
The graph edges are classified into tree, back, cross, and forward edges using
the corresponding depth-first forest~\cite{CLRS}. Cross and forward edges of
$G_1$, $\cfgarrow_{\CFKw} = \{(2,8)\}$, are removed on \pref{li:wto-rm}. Because
$\texttt{lca}_{\DKw}(2,8) = 2$, the removed edge $(2,8)$ will be restored in
\pref{li:wto-rs} when $\etch = 2$.
It is restored as $(2,7)$, because the disjoint set $\{8\}$ would have already
been merged with $\{7\}$ on \pref{li:wto-merge} when $\etch = 7$, making
$\repKw(8)$ to be $7$ when $\etch = 2$.

The for-loop on \pref{li:wto-loop} visits nodes in $V$ in a descending
\texttt{DFN}: from $9$ to $1$. Calling \genFMKw{$\etch$} on \pref{li:wto-scc}
generates $\instr[\etch]$, an $\FM$ instruction for $\etch$. When $\etch = 9$,
because the SCC whose entry is $9$ is trivial, $\fmexec\ 9$ is generated in
\pref{li:wto-exec}. When $\etch = 3$, the SCC whose entry is $3$ is non-trivial,
with the entries of its nested SCCs, $N_\etch = \{4, 6\}$. These entries are
visited in a topological order (descending \texttt{postDFN}), $4, 6$, and their
instructions are connected on \pref{li:wto-repeatbody} to generate $\fmrepeat\
3\ \texttt{[}\instr[4] \fmseq \instr[6]\texttt{]}$ on \pref{li:wto-repeat}.
Visiting the nodes in a descending \texttt{DFN} guarantees the instruction of
nested SCCs to be present, and removing the cross and forward edges ensures each
SCC to have a single entry.
\pref{tab:effex} shows some relevant steps and values within \genFMKw.

Finally, calling \maxSCCKw on \pref{li:wto-mscc} connects the instructions of
entries of outermost SCCs, which is detected by the boolean expression
$\repKw(v) = v$, in a topological order (descending \texttt{postDFN}) to
generate the final $\FM$ program. For the given example, it visits the nodes in
the order of $1$, $2$, $3$, $7$, and $9$, correctly generating the $\FM$ program
on \pref{li:wto-final}.

Due to $2 \cfgarrow' 3$ and $2 \cfgarrow' 7$, $\dpostopt[2]$ is set to $3$ and
then to $7$ on \pref{li:wto-prec1}. Due to $5 \cfgarrow_{\BKw} 4$ and $5
\cfgarrow_{\BKw} 3$, $\dpostopt[5]$ is set to $4$ and then to $3$ in
\pref{li:wto-nest}. $\achkopt[4]$ is set to 3, as $\repKw(4) = 3$ in
\pref{li:wto-achk}. $\TKw[7]$ is set to $2$ on \pref{li:wto-Tprec1}, and
$\dpostlopt[7]$ is set to $\{2\}$ on \pref{li:wto-dpostl}. $\TKw[5]$ is set to
$4$ and then to $3$ on \pref{li:wto-Tnest}, making $\dpostlopt[5]$ to be $\{3,
4, 5\}$. Because $\repKw(4) = 3$, $\dprelopt[4]$ is set to $\{3\}$ in
\pref{li:wto-dprel}.
\qef
\end{example}

\begin{table}[t]
  \caption{Relevant steps and values within \algwto{} when applied to graph
  $G_1$ of \pref{exa:effex}}
  \label{tab:effex}
  \let\center\empty
  \let\endcenter\relax
  \centering
  \small
\begin{tabular}{l|l|l}
  Major iteration & $\etch = 4$ & $\etch = 3$ \\
\midrule
  \pref{li:wto-repeatbody} & $\instr[5]$ & $\instr[4] \fmseq \instr[6]$ \\
\hline
  \pref{li:wto-repeat} & $\fmrepeat\ 4\ \texttt{[}\fmexec\ 5\texttt{]}$ & $\fmrepeat\ 3\ \texttt{[}\fmrepeat\ 4\ \texttt{[}\fmexec\ 5\texttt{]} \fmseq \fmexec\ 6 \texttt{]}$ \\
\hline
  \pref{li:wto-prec2} & $\dpostopt[4] = 5$ & $\dpostopt[4] = 6$, $\dpostopt[3] = 4$ \\
\hline
  \pref{li:wto-Tprec2} & $\TKw[4] = 4$ & $\TKw[4] = 4$, $\TKw[3] = 3$ \\
\hline
  \pref{li:wto-nest} & $\dpostopt[5] = \TKw[5] = 4$ & $\dpostopt[6] = \TKw[6] = \dpostopt[5] = \TKw[5] = 3$ \\
\hline
  \pref{li:wto-merge} & Sets $\{4\}$, $\{5\}$ merged. & Sets $\{3\}$, $\{4,5\}$, $\{6\}$ merged.\\
\end{tabular}
  \vspace{-1ex}
\end{table}

The proofs of the following theorems are in \pref{app:appendix-alg}:
\begin{restatable}{theorem}{ThmAlgCorrect}
  \label{thm:algcorrect}
  \algwto{} correctly computes $\memconfigopt$, defined in \pref{sec:Algorithm}.
\end{restatable}
\begin{restatable}{theorem}{ThmAlgFast}
  \label{thm:algfast}
  Running time of \algwto{} is almost-linear.
\end{restatable}

\section{Implementation}
\label{sec:Implementation}

We have implemented our approach in a tool called $\mikos$, which extends NASA's
$\ikos$ \cite{ikos2014}, a WTO-based abstract-interpreter for C/C++. $\mikos$
inherits all abstract domains and widening-narrowing strategies from $\ikos$. It
includes the localized narrowing strategy~\cite{amatoSAS2013} that intertwines
the increasing and decreasing sequences.

\subsubsubsection{Abstract domains in $\ikos$.}
$\ikos$ uses the state-of-the-art implementations of abstract domains comparable
to those used in industrial abstract interpreters such as Astr\'ee. In
particular, $\ikos$ implements the interval abstract domain~\cite{kn:CC77} using
functional data-structures based on Patricia Trees~\cite{okasaki1998fast}.
Astr\'ee implements intervals using OCaml's map data structure that uses
balanced trees~\cite[Section 6.2]{DBLP:conf/birthday/BlanchetCCFMMMR02}. As
shown in \cite[Section 5]{okasaki1998fast}, the Patricia Trees used by $\ikos$
are more efficient when you have to merge data structures, which is required
often during abstract interpretation. Also, $\ikos$ uses memory-efficient
variable packing Difference Bound Matrix (DBM) relational abstract
domain~\cite{GangeNSSS16:VMCAI2016}, similar to the variable packing relational
domains employed by Astr\'ee~\cite[Section
3.3.2]{DBLP:journals/sigsoft/BertraneCCFMMR11}.

\subsubsubsection{Interprocedural analysis in $\ikos$.}
$\ikos$ implements context-sensitive interprocedural analysis by means of
dynamic inlining, much like the semantic expansion of function bodies in
Astr\'ee~\cite[Section 5]{DBLP:conf/esop/CousotCFMMMR05}: at a function call,
formal and actual parameters are matched, the callee is analyzed, and the return
value at the call site is updated after the callee returns; a function pointer
is resolved to a set of callees and the results for each call are joined;
$\ikos$ returns top for a callee when a cycle is found in this dynamic call
chain. To prevent running the entire interprocedural analysis again at the
assertion checking phase, invariants at exits of the callees are additionally
cached during the fixpoint computation.

\subsubsubsection{Interprocedural extension of $\mikos$.}
Although the description of our iteration strategy focused on intraprocedural
analysis, it can be extended to interprocedural analysis as follows. Suppose
there is a call to function \texttt{f1} from a basic block contained in
component $C$. Any checks in this call to \texttt{f1} must be deferred until we
know that the component $C$ has stabilized. Furthermore, if function \texttt{f1}
calls the function \texttt{f2}, then the checks in \texttt{f2} must also be
deferred until $C$ converges. In general, checks corresponding to a function
call $f$ must be deferred until the maximal component containing the call is
stabilized.

When the analysis of callee returns in $\mikos$, only \pre{} values for the
deferred checks remain. They are deallocated when the checks are performed or
when the component containing the call is reiterated.
\section{Experimental Evaluation}
\label{sec:evaluation}

The experiments in this section were designed to answer the following
questions:
\begin{itemize}%
  \item[] \RQ{0} \textbf{[Accuracy]} Does $\mikos$ (\pref{sec:Implementation})
  have the same analysis results as $\ikos$?
  \item[] \RQ{1} \textbf{[Memory footprint]} How does the memory footprint of
  $\mikos$ compare to that of \ikos?
  \item[]\RQ{2} \textbf{[Runtime]} How does the runtime of $\mikos$ compare to
  that of \ikos?
\end{itemize}

\subsubsubsection{Experimental setup}
All experiments were run on Amazon EC2 r5.2xlarge instances (64 GiB memory, 8
vCPUs, 4 physical cores), which use Intel Xeon Platinum 8175M processors.
Processors have L1, L2, and L3 caches of sizes 1.5 MiB (data: 0.75 MiB,
instruction: 0.75 MiB), 24 MiB, and 33 MiB, respectively. Linux kernel version
4.15.0-1051-aws was used, and gcc 7.4.0 was used to compile both $\mikos$ and
$\ikos$. Dedicated EC2 instances and BenchExec~\cite{Beyer2019} were used to
improve reliability of the results. Time and space limit were set to
\timelimit{} and \memlimit{}, respectively.
The experiments can be reproduced using \url{https://github.com/95616ARG/mikos_sas2020}.

\subsubsubsection{Benchmarks}
\label{sec:BenchmarkFilter}
We evaluated $\mikos$ on two tasks that represent different client applications
of abstract interpretation, each using different benchmarks described in
Sections~\ref{sec:T1} and \ref{sec:T2}. In both tasks, we excluded benchmarks
that did not complete in \emph{both} $\ikos$ and $\mikos$ given the time and
space budget.
There were no benchmarks for which $\ikos$ succeeded but $\mikos$ failed to complete.
Benchmarks for which $\ikos$ took less
than \filtertime{} were also excluded.
Measurements for benchmarks that took less than \filtertime{} are summarized in
\ifarxiv
\pref{app:appendix-eval}.
\else
our technical report.
\fi

\subsubsubsection{Metrics}
To answer RQ1, we define and use \emph{memory reduction ratio (MRR)}:
\vspace{-2ex}
\begin{equation}
  \label{eq:MRR}
  \text{MRR} \eqdef \text{Memory footprint of $\mikos$ } / \text{ Memory footprint of $\ikos$}
  \vspace{-1ex}
\end{equation}
The smaller the MRR, the greater reduction in peak-memory usage in $\mikos$. If
MRR is less than 1, $\mikos$ has smaller memory footprint than $\ikos$.

For RQ2, we report the \emph{speedup}, which is defined as below:
\vspace{-2ex}
\begin{equation}
  \label{eq:speedup}
  \text{Speedup} \eqdef \text{Runtime of $\ikos$ } / \text{ Runtime of $\mikos$}
  \vspace{-1ex}
\end{equation}
The larger the speedup, the greater reduction in runtime in $\mikos$. If speedup
is greater than 1, $\mikos$ is faster than $\ikos$.

\subsubsection{RQ0: Accuracy of $\mikos$}
As a sanity check for our theoretical results, we experimentally validated
\pref{thm:full} by comparing the analysis results reported by $\ikos$ and
$\mikos$. $\mikos$ used a valid memory configuration, reporting the same
analysis results as $\ikos$. Recall that \pref{thm:full} also proves
that the fixpoint computation in $\mikos$ is memory-optimal (, it results
in minimum memory footprint).

\subsection{Task T1: Verifying user-provided assertions}
\label{sec:T1}
\begin{figure}[t]
  \centering
  \begin{subfigure}[t]{0.48\textwidth}
    \centering
    \includegraphics[width=\textwidth]{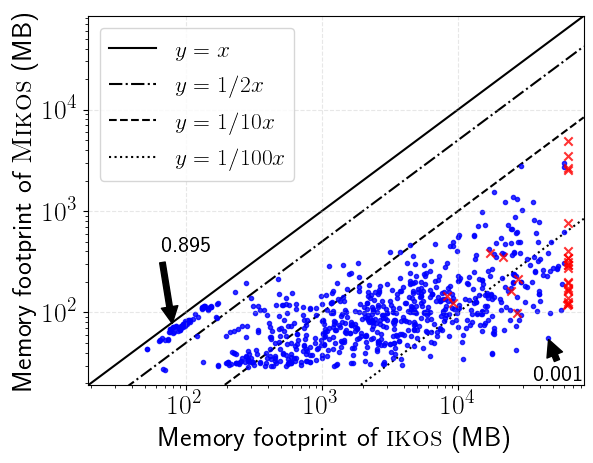}
    \caption{Min MRR: \svcminmem{}. Max MRR: \svcmaxmem{}. Geometric means:
    (i)~\svcgeomeanok{} (when \noncomplete{}s are ignored), (ii)~\svcgeomean{}
    (when measurements until timeout/spaceout are used for \noncomplete{}s).
    \svcfilterednok{} non-completions in $\ikos$.}
    \label{fig:mscat-svc}
  \end{subfigure}
  \hfill
  \begin{subfigure}[t]{0.48\textwidth}
    \centering
    \includegraphics[width=\textwidth]{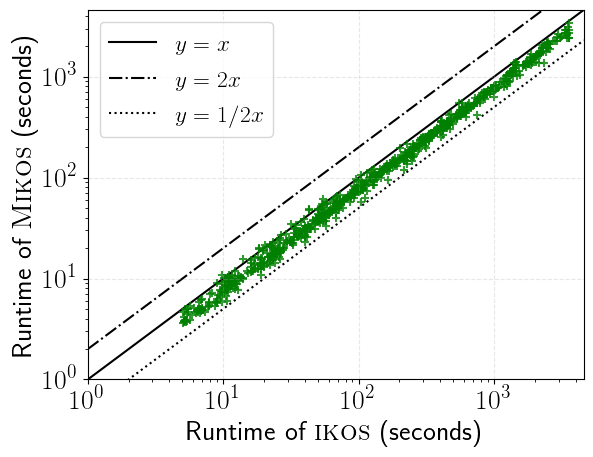}
    \caption{Min speedup: \svcspeedupmin{}. Max speedup: \svcspeedupmax{}.
      Geometric mean: \svcspeedupavg{}. Note that \noncomplete{}s are ignored as
      they space out fast in $\ikos$ compared to in $\mikos$ where they
      complete.}
    \label{fig:tscat-svc}
  \end{subfigure}
  \vspace{-1.3ex}
  \caption{\textbf{Task T1.} Log-log scatter plots of \protect\subref{fig:mscat-svc}
  memory footprint and \protect\subref{fig:tscat-svc} runtime of $\ikos$ and
  $\mikos$, with an hour timeout and \memlimit{} spaceout. Benchmarks that did
  not complete in $\ikos$ are marked \textcolor{red}{$\times$}. All
  \textcolor{red}{$\times$}s completed in $\mikos$. Benchmarks below $y = x$
  required less memory or runtime in $\mikos$.}
  \label{fig:svcplot}
  \vspace{-1ex}
\end{figure}

\subsubsubsection{Benchmarks}
For Task T1, we selected all \svctotal{} benchmarks from DeviceDriversLinux64,
ControlFlow, and Loops categories of SV-COMP 2019 \cite{svcomp}. These
categories are well suited for numerical analysis, and have been used in recent
works~\cite{DBLP:conf/cav/SinghPV18,DBLP:journals/pacmpl/SinghPV18,DBLP:journals/pacmpl/KimVT20}. From
these benchmarks, we removed \svcmikosto{} benchmarks that timed out in both
$\mikos$ and $\ikos$, and \svcikosfive{} benchmarks that took less than 5
seconds in $\ikos$. That left us with \textbf{\svcfiltered{}} SV-COMP 2019
benchmarks.

\subsubsubsection{Abstract domain}
Task T1 used the reduced product of Difference Bound Matrix (DBM) with variable
packing~\cite{GangeNSSS16:VMCAI2016} and congruence~\cite{Granger:IJCM1989}.
This domain is much richer and more expressive than the interval domain used in
task T2.

\subsubsubsection{Task}
Task T1 consists of using the results of interprocedural fixpoint computation to
prove user-provided assertions in the SV-COMP benchmarks. Each benchmark
typically has one assertion to prove.

\ifarxiv
\begin{figure}[t]
  \centering
  \begin{subfigure}[b]{0.49\textwidth}
    \centering
    \includegraphics[width=\textwidth]{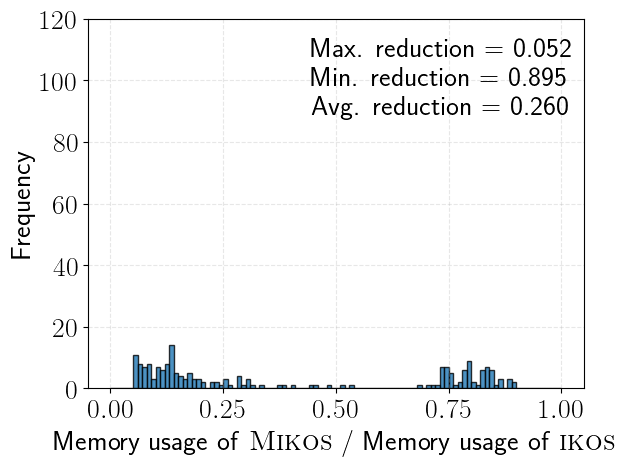}
    \caption[]%
    {{\small 0\% -- 25\% ($\svchistal$ MB -- $\svchistau$ MB)}}
    \label{fig:bottom-svc}
  \end{subfigure}
  \begin{subfigure}[b]{0.49\textwidth}
    \centering
    \includegraphics[width=\textwidth]{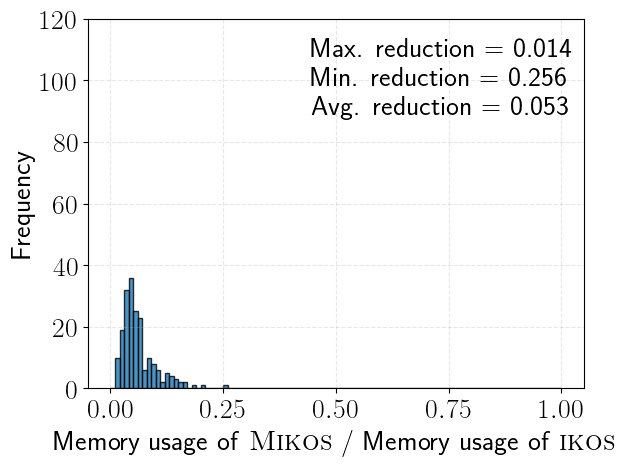}
    \caption[]%
    {{\small 25\% -- 50\% ($\svchistbl$ MB -- $\svchistbu$ MB)}}
  \end{subfigure}
  \begin{subfigure}[b]{0.49\textwidth}
    \centering
    \includegraphics[width=\textwidth]{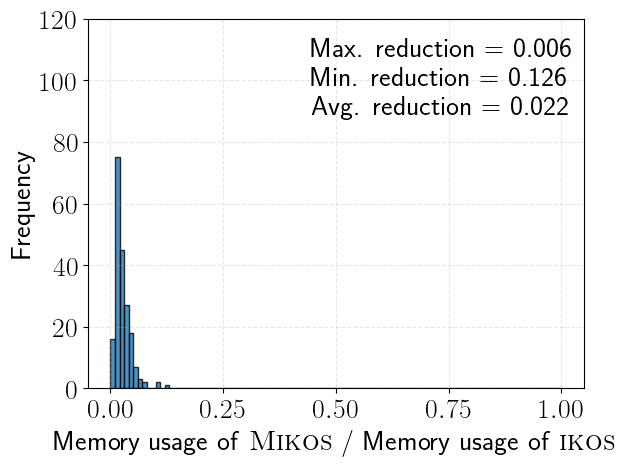}
    \caption[]%
    {{\small 50\% -- 75\% ($\svchistcl$ MB -- $\svchistcu$ MB)}}
  \end{subfigure}
  \begin{subfigure}[b]{0.49\textwidth}
    \centering
    \includegraphics[width=\textwidth]{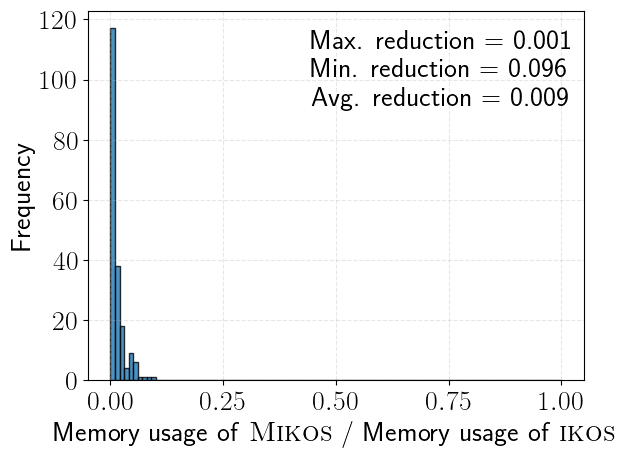}
    \caption[]%
    {{\small 75\% -- 100\% ($\svchistdl$ MB -- $\svchistdu$ MB)}}
  \end{subfigure}
  \caption{Histograms of MRR (\pref{eq:MRR}) in task T1 for different ranges.
  \pref{fig:bottom-svc} shows the distribution of benchmarks that used from
  $\svchistal$ MB to $\svchistau$ MB in $\ikos$. They are the bottom 25\% in
  terms of the memory footprint in $\ikos$. The distribution significantly
  tended toward a smaller MRR in the upper range.}
  \label{fig:svchistplot}
\end{figure}
\fi

\subsubsubsection{RQ1: Memory footprint of $\mikos$ compared to $\ikos$}
\pref{fig:mscat-svc} shows the measured memory footprints in a log-log scatter
plot. For Task T1, the MRR (\pref{eq:MRR}) ranged from \svcminmem{} to
\svcmaxmem{}. That is, the memory footprint decreased to 0.1\% in the best case.
For all benchmarks, $\mikos$ had smaller memory footprint than $\ikos$: MRR was
less than 1 for all benchmarks, with all points below the $y = x$ line in
\pref{fig:mscat-svc}. On average, $\mikos$ required only 4.1\% of the memory
required by $\ikos$, with an MRR \svcgeomean{} as the geometric mean.

As \pref{fig:mscat-svc} shows, reduction in memory tended to be greater as the
memory footprint in the baseline $\ikos$ grew. For the top 25\% benchmarks with
largest memory footprint in $\ikos$, the geometric mean of MRRs was
\svcgeomeanq{}. This trend is further confirmed by the histograms in
\pref{fig:svchistplot}. While a similar trend was observed in task T2, the trend
was significantly stronger in task T1.
\ifarxiv
\pref{tab:rq1-svc} in \pref{app:appendix-eval} lists \textbf{RQ1} results for
specific benchmarks.
\else
Our technical reports has more detailed numbers.
\fi

\subsubsubsection{RQ2: Runtime of $\mikos$ compared to $\ikos$}
\pref{fig:tscat-svc} shows the measured runtime in a log-log scatter plot. We
measured both the speedup (\pref{eq:speedup}) and the difference in the
runtimes. For fair comparison, we excluded 29 benchmarks that did not complete
in $\ikos$. This left us with 755 SV-COMP 2019 benchmarks. Out of these
\svcfilteredok{} benchmarks, \svcspeedup{} benchmarks had speedup $>1$. The
speedup ranged from \svcspeedupmin{} to \svcspeedupmax{}, with geometric mean of
\svcspeedupavg{}. The difference in runtimes (runtime of $\ikos$ $-$ runtime of
$\mikos$) ranged from $\svcspeedupabsmin$ to $\svcspeedupabsmax$, with
arithmetic mean of $\svcspeedupabsavg$.
\ifarxiv
\pref{tab:rq2-svc} in \pref{app:appendix-eval} lists \textbf{RQ2}
results for specific benchmarks.
\else
Our technical report has more detailed numbers.
\fi
\subsection{Task T2: Proving absence of buffer overflows}
\label{sec:T2}

\begin{figure}[t]
  \centering
  \begin{subfigure}[t]{0.48\textwidth}
    \centering
    \includegraphics[width=\textwidth]{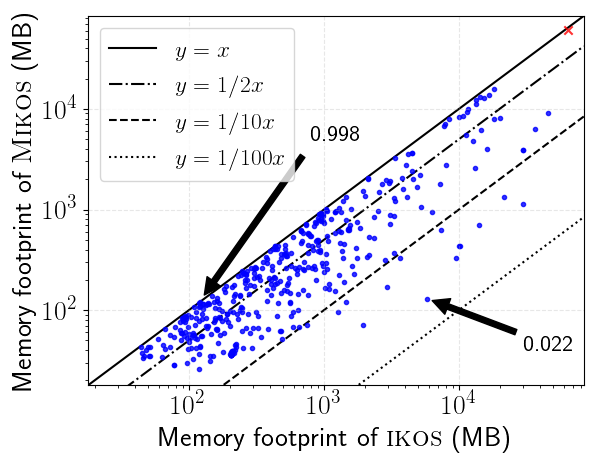}
    \caption{Min MRR: \ossminmem{}. Max MRR: \ossmaxmem{}. Geometric means:
    (i)~\ossgeomeanok{} (when \noncomplete{}s are ignored), (ii)~\ossgeomean{}
    (when measurements until timeout/spaceout are used for \noncomplete{}s).
    \ossfilterednok{} non-completions in $\ikos$.}
    \label{fig:mscat-oss}
  \end{subfigure}
  \hfill
  \begin{subfigure}[t]{0.48\textwidth}
    \centering
    \includegraphics[width=\textwidth]{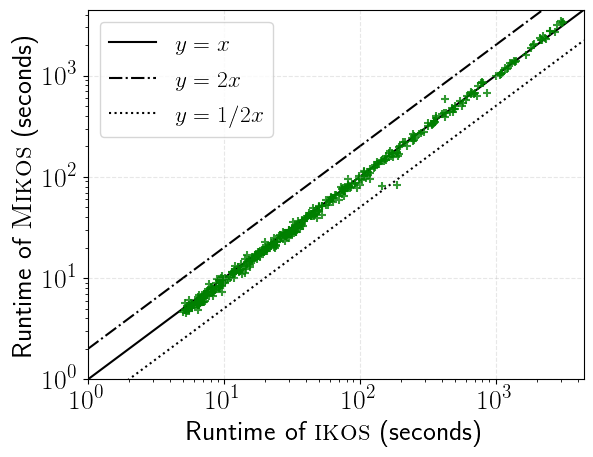}
    \caption{Min speedup: \ossspeedupmin{}. Max speedup: \ossspeedupmax{}.
      Geometric mean: \ossspeedupavg{}. Note that \noncomplete{}s are ignored as
      they space out fast in $\ikos$ compared to in $\mikos$ where they
      complete.}
    \label{fig:tscat-oss}
  \end{subfigure}
  \vspace{-1.3ex}
  \caption{\textbf{Task T2.} Log-log scatter plots of \protect\subref{fig:mscat-oss}
  memory footprint and \protect\subref{fig:tscat-oss} runtime of $\ikos$ and
  $\mikos$, with an hour timeout and \memlimit{} spaceout. Benchmarks that did
  not complete in $\ikos$ are marked \textcolor{red}{$\times$}. All
  \textcolor{red}{$\times$}s completed in $\mikos$. Benchmarks below $y = x$
  required less memory or runtime in $\mikos$.}
  \label{fig:ossplot}
  \vspace{-1ex}
\end{figure}

\subsubsubsection{Benchmarks}
For Task T2, we selected all \osstotal{} programs from the official Arch Linux
core packages that are primarily written in C and whose LLVM bitcodes are
obtainable by gllvm~\cite{gllvm}. These include, but are not limited to,
\texttt{coreutils}, \texttt{dhcp}, \texttt{gnupg}, \texttt{inetutils},
\texttt{iproute}, \texttt{nmap}, \texttt{openssh}, \texttt{vim}, etc. From these
benchmarks, we removed \ossmikosto{} benchmarks that timed out and \ossmikosso{}
benchmarks that spaced out in both $\mikos$ and $\ikos$. Also, \ossikosfive{}
benchmarks that took less than 5 seconds in $\ikos$ were removed. That left us
with \textbf{\ossfiltered{}} open-source benchmarks.

\subsubsubsection{Abstract domain}
Task T2 used the interval abstract domain~\cite{kn:CC77}. Using a richer domain
like DBM caused $\ikos$ and $\mikos$ to timeout on most benchmarks.

\subsubsubsection{Task}
Task T2 consists of using the results of interprocedural fixpoint computation to
prove the safety of buffer accesses. In this task, most program points had
checks.

\subsubsubsection{RQ1: Memory footprint of $\mikos$ compared to $\ikos$}
\pref{fig:mscat-oss} shows the measured memory footprints in a log-log scatter
plot. For Task T2, MRR (\pref{eq:MRR}) ranged from \ossminmem{} to \ossmaxmem{}.
That is, the memory footprint decreased to 2.2\% in the best case. For all
benchmarks, $\mikos$ had smaller memory footprint than $\ikos$: MRR was less
than 1 for all benchmarks, with all points below the $y = x$ line in
\pref{fig:mscat-oss}. On average, $\mikos$'s memory footprint was less than half
of that of $\ikos$, with an MRR \ossgeomean{} as the geometric mean.
\ifarxiv
\pref{tab:rq1-oss} in \pref{app:appendix-eval} lists \textbf{RQ1} results for
specific benchmarks.
\else
Our technical reports has more detailed numbers.
\fi

\ifarxiv
\begin{figure}[t]
  \centering
  \begin{subfigure}[b]{0.49\textwidth}
    \centering
    \includegraphics[width=\textwidth]{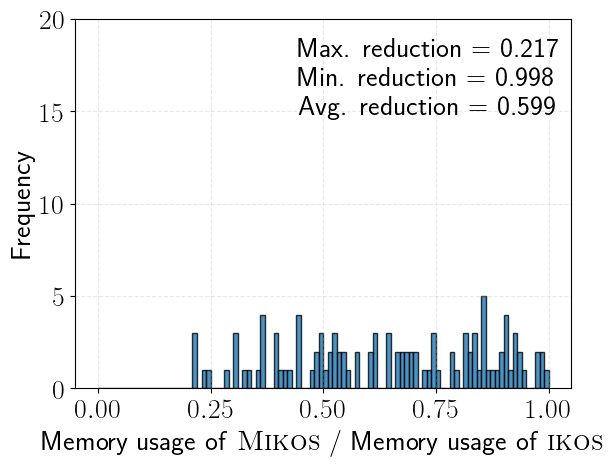}
    \caption[]%
    {{\small 0\% -- 25\% ($\osshistal$ MB -- $\osshistau$ MB)}}
    \label{fig:bottom-oss}
  \end{subfigure}
  \begin{subfigure}[b]{0.49\textwidth}
    \centering
    \includegraphics[width=\textwidth]{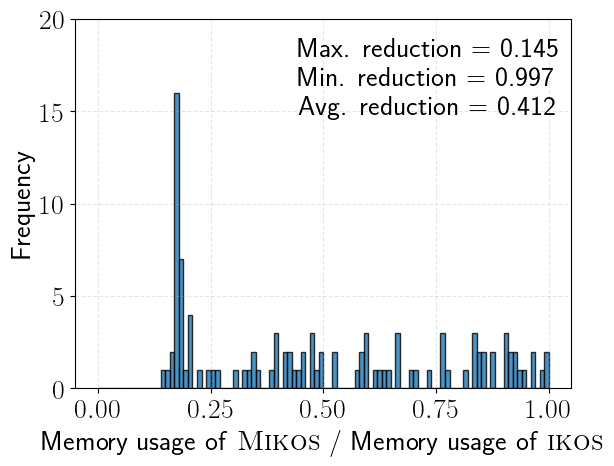}
    \caption[]%
    {{\small 25\% -- 50\% ($\osshistbl$ MB -- $\osshistbu$ MB)}}
  \end{subfigure}
  \begin{subfigure}[b]{0.49\textwidth}
    \centering
    \includegraphics[width=\textwidth]{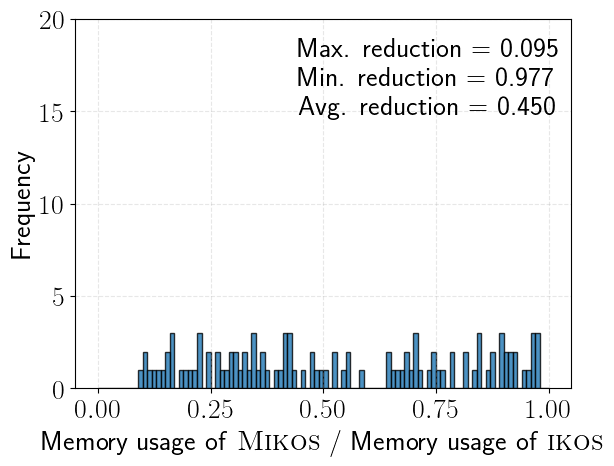}
    \caption[]%
    {{\small 50\% -- 75\% ($\osshistcl$ MB -- $\osshistcu$ MB)}}
  \end{subfigure}
  \begin{subfigure}[b]{0.49\textwidth}
    \centering
    \includegraphics[width=\textwidth]{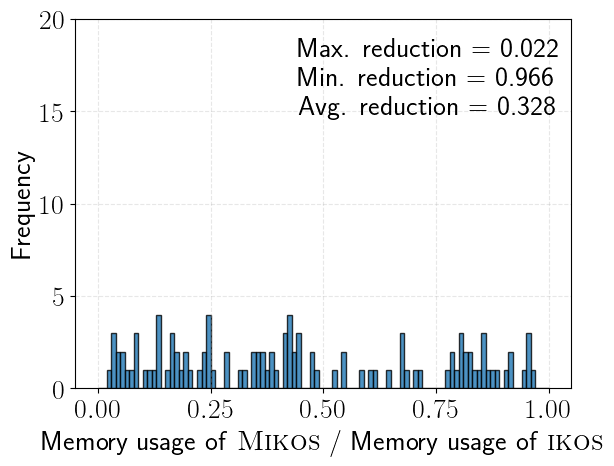}
    \caption[]%
    {{\small 75\% -- 100\% ($\osshistdl$ MB -- $\osshistdu$ MB)}}
  \end{subfigure}
  \caption{Histograms of MRR (\pref{eq:MRR}) in task T2 for different ranges.
  \pref{fig:bottom-oss} shows the distribution of benchmarks that used from
  $\osshistal$ MB to $\osshistau$ MB in $\ikos$. They are the bottom 25\% in
  terms of the memory footprint in $\ikos$. The distribution slightly tended
  toward a smaller MRR in the upper range.}
  \label{fig:osshistplot}
\end{figure}
\fi

\subsubsubsection{RQ2: Runtime of $\mikos$ compared to $\ikos$}
\pref{fig:tscat-oss} shows the measured runtime in a log-log scatter plot. We
measured both the speedup (\pref{eq:speedup}) and the difference in the
runtimes. For fair comparison, we excluded 1 benchmark that did not complete in
$\ikos$. This left us with 425 open-source benchmarks. Out of these
\ossfilteredok{} benchmarks, \ossspeedup{} benchmarks had speedup $>1$. The
speedup ranged from \ossspeedupmin{} to \ossspeedupmax{}, with geometric mean of
\ossspeedupavg{}. The difference in runtimes (runtime of $\ikos$ $-$ runtime of
$\mikos$) ranged from $\ossspeedupabsmin$ to $\ossspeedupabsmax$, with
arithmetic mean of $\ossspeedupabsavg$.
\ifarxiv
\pref{tab:rq2-oss} in \pref{app:appendix-eval} lists \textbf{RQ2}
results for specific benchmarks.
\else
Our technical report has more detailed numbers.
\fi

\vspace{-1ex}
\section{Related Work}
\label{sec:Related}
\vspace{-1ex}

Abstract interpretation has a long history of designing time and memory
efficient algorithms for specific abstract domains, which exploit variable
packing and clustering and sparse constraints
\cite{DBLP:journals/pacmpl/SinghPV18,DBLP:conf/cav/SinghPV18,DBLP:conf/popl/SinghPV17,DBLP:conf/pldi/SinghPV15,DBLP:conf/sas/HeoOY16,DBLP:conf/sas/GangeNSSS16,DBLP:conf/aplas/ChawdharyK17,DBLP:journals/fmsd/HalbwachsMG06}.
Often these techniques represent a trade-off between precision and performance
of the analysis. Nonetheless, such techniques are orthogonal to the
abstract-domain agnostic approach discussed in this paper. Approaches for
improving precision via sophisticated widening and narrowing strategies
\cite{halbwachsSAS2012,DBLP:journals/scp/AmatoSSAV16,DBLP:conf/birthday/ApinisSV16}
are also orthogonal to our memory-efficient iteration strategy.
$\mikos$ inherits the interleaved widening-narrowing strategy implemented in the
baseline $\ikos$ abstract interpreter.

As noted in \pref{sec:Introduction}, Bourdoncle's
approach~\cite{bourdoncle1993efficient} is used in many industrial and academic
abstract interpreters~\cite{ikos2014,crab,sparta,codehawk,DBLP:conf/nfm/CalcagnoD11}.
Thus, improving memory efficiency of WTO-based exploration is of great
applicability to real-world static analysis. Astr\'ee is one of the few, if not
only, industrial abstract interpreters that does not use WTO exploration, because
it assumes that programs do not have gotos and recursion~\cite[Section
2.1]{DBLP:conf/birthday/BlanchetCCFMMMR02}, and is targeted towards a specific
class of embedded C code~\cite[Section
3.2]{DBLP:journals/sigsoft/BertraneCCFMMR11}. Such restrictions makes is easier
to compute when an abstract value will not be used anymore by naturally
following the abstract syntax tree~\cite[Section
3.4.3]{DBLP:journals/ftpl/Mine17}. In contrast, $\mikos$ works for general
programs with goto and recursion, which requires the use of WTO-based
exploration.

Generic fixpoint-computation approaches for improving running time of abstract
interpretation have also been explored
\cite{PLDI:VB2004,monniaux2005parallel,DBLP:journals/pacmpl/KimVT20}. Most
recently, Kim et al.~\cite{DBLP:journals/pacmpl/KimVT20} present the notion of
weak partial order (WPO), which generalizes the notion of WTO that is used in
this paper. Kim et al.\ describe a parallel fixpoint algorithm that exploits
maximal parallelism while computing the same fixpoint as the WTO-based
algorithm. Reasoning about correctness of concurrent algorithms is complex;
hence, we decided to investigate an optimal memory management scheme in the
sequential setting first. However, we believe it would be possible to extend our
WTO-based result to one that uses WPO.

The nesting relation described in \pref{sec:Algorithm} is closely related to the
notion of Loop Nesting Forest~\cite{ramalingamTOPLAS1999,ramalingamTOPLAS2002},
as observed in Kim~et~al.~\cite{DBLP:journals/pacmpl/KimVT20}. The almost-linear
time algorithm \algwto{} is an adaptation of LNF construction algorithm by
Ramalingam~\cite{ramalingamTOPLAS1999}. The $\foo$ operation in
\pref{sec:Algorithm} is similar to the outermost-loop-excluding (OLE) operator
introduced by Rastello~\cite[Section 2.4.4]{rastello:hal-00761555}.

Seidl et al.~\cite{DBLP:conf/ppdp/SeidlV18} present time and space improvements
to a generic fixpoint solver, which is closest in spirit to the problem
discussed in this paper. For improving space efficiency, their approach
recomputes values during fixpoint computation, and does not prove optimality,
unlike our approach. However, the setting discussed in their work is also more
generic compared to ours; we assume a static dependency graph for the equation
system.

Abstract interpreters such as
Astr\'ee~\cite{DBLP:conf/birthday/BlanchetCCFMMMR02} and
CodeHawk~\cite{codehawk} are implemented in OCaml, which provides a garbage
collector. However, merely using a reference counting garbage collector will not
reduce peak memory usage of fixpoint computation. For instance, the reference
count of $\pre[u]$ can be decreased to zero only after the final check/assert
that uses $\pre[u]$. If the checks are all conducted at the end of the analysis
(as is currently done in prior tools), then using a reference counting garbage
collector will not reduce peak memory usage. In contrast, our approach lifts the
checks as early as possible enabling the analysis to free the
abstract values as early as possible.

Symbolic approaches for applying abstract transformers during fixpoint
computation
\cite{DBLP:journals/entcs/HenryMM12,DBLP:conf/vmcai/RepsSY04,DBLP:conf/popl/LiAKGC14,DBLP:conf/vmcai/RepsT16,DBLP:journals/entcs/ThakurLLR15,DBLP:conf/sas/ThakurER12,DBLP:conf/cav/ThakurR12}
allow the entire loop body to be encoded as a single formula. This might appear
to obviate the need for $\pre$ and $\post$ values for individual basic blocks
within the loop; by storing the $\pre$ value only at the header, such a symbolic
approach might appear to reduce the memory footprint. First, this scenario does
not account for the fact that $\pre$ values need to be computed and stored if
basic blocks in the loop have checks. Note that if there are no checks within
the loop body, then our approach would also only store the $\pre$ value at the
loop header. Second, such symbolic approaches only perform intraprocedural
analysis~\cite{DBLP:journals/entcs/HenryMM12}; additional abstract values would
need to be stored depending on how function calls are handled in interprocedural
analysis. Third, due to the use of SMT solvers in such symbolic approaches, the
memory footprint might not necessarily reduce, but might increase if one takes
into account the memory used by the SMT solver.

Sparse analysis \cite{DBLP:conf/pldi/OhHLLY12,DBLP:journals/toplas/OhHLLPKY14}
and database-backed analysis \cite{WeissRL:ICSE2015} improve the memory cost of
static analysis. For specific classes of static analysis such as the IFDS
framework \cite{RHSPOPL1995}, there have been approaches for improving the time
and memory efficiency
\cite{DBLP:conf/pldi/Bodden12a,DBLP:conf/cc/NaeemLR10,wangASPLOS2017,DBLP:conf/ipps/0002GJWHWL19}.

\vspace{-1ex}
\section{Conclusion}
\vspace{-1ex}
\label{sec:Conclusion}

This paper presented an approach for memory-efficient abstract interpretation
that is agnostic to the abstract domain used. Our approach is memory-optimal and
produces the same result as Bourdoncle's approach without sacrificing time
efficiency. We extended the notion of iteration strategy to intelligently
deallocate abstract values and perform assertion checks during fixpoint
computation. We provided an almost-linear time algorithm that constructs this
iteration strategy. We implemented our approach in a tool called $\mikos$, which
extended the abstract interpreter $\ikos$. Despite the use of state-of-the-art
implementation of abstract domains, $\ikos$ had a large memory footprint on two
analysis tasks. $\mikos$ was shown to effectively reduce it. When verifying
user-provided assertions in SV-COMP 2019 benchmarks, $\mikos$ showed a decrease
in peak-memory usage to $\svcgeomeanpercentage$\% ($\svcgeomeantimes\times$) on
average compared to $\ikos$. When performing interprocedural buffer-overflow
analysis of open-source programs, $\mikos$ showed a decrease in peak-memory
usage to $\ossgeomeanpercentage$\% ($\ossgeomeantimes\times$) on average
compared to $\ikos$.
\bibliographystyle{splncs04}
\bibliography{main}

\begin{thebibliography}{10}
\providecommand{\url}[1]{\texttt{#1}}
\providecommand{\urlprefix}{URL }
\providecommand{\doi}[1]{https://doi.org/#1}

\bibitem{amatoSAS2013}
Amato, G., Scozzari, F.: Localizing widening and narrowing. In: Static Analysis
  - 20th International Symposium, {SAS} 2013, Seattle, WA, USA, June 20-22,
  2013. Proceedings. pp. 25--42 (2013). \doi{10.1007/978-3-642-38856-9\_4}

\bibitem{DBLP:journals/scp/AmatoSSAV16}
Amato, G., Scozzari, F., Seidl, H., Apinis, K., Vojdani, V.: Efficiently
  intertwining widening and narrowing. Sci. Comput. Program.  \textbf{120},
  1--24 (2016). \doi{10.1016/j.scico.2015.12.005}

\bibitem{DBLP:conf/birthday/ApinisSV16}
Apinis, K., Seidl, H., Vojdani, V.: Enhancing top-down solving with widening
  and narrowing. In: Probst, C.W., Hankin, C., Hansen, R.R. (eds.) Semantics,
  Logics, and Calculi - Essays Dedicated to Hanne Riis Nielson and Flemming
  Nielson on the Occasion of Their 60th Birthdays. Lecture Notes in Computer
  Science, vol.~9560, pp. 272--288. Springer (2016).
  \doi{10.1007/978-3-319-27810-0\_14}

\bibitem{DBLP:journals/scp/BagnaraHZ08}
Bagnara, R., Hill, P.M., Zaffanella, E.: The parma polyhedra library: Toward a
  complete set of numerical abstractions for the analysis and verification of
  hardware and software systems. Sci. Comput. Program.  \textbf{72}(1-2),
  3--21 (2008). \doi{10.1016/j.scico.2007.08.001}

\bibitem{DBLP:journals/sigsoft/BertraneCCFMMR11}
Bertrane, J., Cousot, P., Cousot, R., Feret, J., Mauborgne, L., Min{\'{e}}, A.,
  Rival, X.: Static analysis by abstract interpretation of embedded critical
  software. {ACM} {SIGSOFT} Software Engineering Notes  \textbf{36}(1), ~1--8
  (2011). \doi{10.1145/1921532.1921553}

\bibitem{svcomp}
Beyer, D.: Automatic verification of {C} and java programs: {SV-COMP} 2019. In:
  Tools and Algorithms for the Construction and Analysis of Systems - 25 Years
  of {TACAS:} TOOLympics, Held as Part of {ETAPS} 2019, Prague, Czech Republic,
  April 6-11, 2019, Proceedings, Part {III}. pp. 133--155 (2019).
  \doi{10.1007/978-3-030-17502-3\_9}

\bibitem{Beyer2019}
Beyer, D., L{\"{o}}we, S., Wendler, P.: Reliable benchmarking: requirements and
  solutions. {STTT}  \textbf{21}(1),  1--29 (2019).
  \doi{10.1007/s10009-017-0469-y}

\bibitem{DBLP:conf/birthday/BlanchetCCFMMMR02}
Blanchet, B., Cousot, P., Cousot, R., Feret, J., Mauborgne, L., Min{\'{e}}, A.,
  Monniaux, D., Rival, X.: Design and implementation of a special-purpose
  static program analyzer for safety-critical real-time embedded software. In:
  Mogensen, T.{\AE}., Schmidt, D.A., Sudborough, I.H. (eds.) The Essence of
  Computation, Complexity, Analysis, Transformation. Essays Dedicated to Neil
  D. Jones [on occasion of his 60th birthday]. Lecture Notes in Computer
  Science, vol.~2566, pp. 85--108. Springer (2002).
  \doi{10.1007/3-540-36377-7\_5}

\bibitem{DBLP:conf/pldi/Bodden12a}
Bodden, E.: Inter-procedural data-flow analysis with {IFDS/IDE} and soot. In:
  Bodden, E., Hendren, L.J., Lam, P., Sherman, E. (eds.) Proceedings of the
  {ACM} {SIGPLAN} International Workshop on State of the Art in Java Program
  analysis, {SOAP} 2012, Beijing, China, June 14, 2012. pp.~3--8. {ACM} (2012).
  \doi{10.1145/2259051.2259052}

\bibitem{bourdoncle1993efficient}
Bourdoncle, F.: Efficient chaotic iteration strategies with widenings. In:
  Formal Methods in Programming and Their Applications, International
  Conference, Akademgorodok, Novosibirsk, Russia, June 28 - July 2, 1993,
  Proceedings. pp. 128--141 (1993). \doi{10.1007/BFb0039704}

\bibitem{ikos2014}
Brat, G., Navas, J.A., Shi, N., Venet, A.: {IKOS:} {A} framework for static
  analysis based on abstract interpretation. In: Software Engineering and
  Formal Methods - 12th International Conference, {SEFM} 2014, Grenoble,
  France, September 1-5, 2014. Proceedings. pp. 271--277 (2014).
  \doi{10.1007/978-3-319-10431-7\_20}

\bibitem{DBLP:conf/nfm/CalcagnoD11}
Calcagno, C., Distefano, D.: Infer: An automatic program verifier for memory
  safety of {C} programs. In: Bobaru, M.G., Havelund, K., Holzmann, G.J.,
  Joshi, R. (eds.) {NASA} Formal Methods - Third International Symposium, {NFM}
  2011, Pasadena, CA, USA, April 18-20, 2011. Proceedings. Lecture Notes in
  Computer Science, vol.~6617, pp. 459--465. Springer (2011).
  \doi{10.1007/978-3-642-20398-5\_33}

\bibitem{DBLP:conf/aplas/ChawdharyK17}
Chawdhary, A., King, A.: Compact difference bound matrices. In: Chang, B.E.
  (ed.) Programming Languages and Systems - 15th Asian Symposium, {APLAS} 2017,
  Suzhou, China, November 27-29, 2017, Proceedings. Lecture Notes in Computer
  Science, vol. 10695, pp. 471--490. Springer (2017).
  \doi{10.1007/978-3-319-71237-6\_23}

\bibitem{CLRS}
Cormen, T.H., Leiserson, C.E., Rivest, R.L., Stein, C.: Introduction to
  Algorithms, 3rd Edition. {MIT} Press (2009)

\bibitem{kn:CC77}
Cousot, P., Cousot, R.: Abstract interpretation: {A} unified lattice model for
  static analysis of programs by construction or approximation of fixpoints.
  In: Conference Record of the Fourth {ACM} Symposium on Principles of
  Programming Languages, Los Angeles, California, USA, January 1977. pp.
  238--252 (1977). \doi{10.1145/512950.512973}

\bibitem{DBLP:conf/esop/CousotCFMMMR05}
Cousot, P., Cousot, R., Feret, J., Mauborgne, L., Min{\'{e}}, A., Monniaux, D.,
  Rival, X.: The astre{\'{e}} analyzer. In: Sagiv, S. (ed.) Programming
  Languages and Systems, 14th European Symposium on Programming,ESOP 2005, Held
  as Part of the Joint European Conferences on Theory and Practice of Software,
  {ETAPS} 2005, Edinburgh, UK, April 4-8, 2005, Proceedings. Lecture Notes in
  Computer Science, vol.~3444, pp. 21--30. Springer (2005).
  \doi{10.1007/978-3-540-31987-0\_3}

\bibitem{sparta}
Facebook: Sparta. \url{https://github.com/facebookincubator/SPARTA} (2020)

\bibitem{GangeNSSS16:VMCAI2016}
Gange, G., Navas, J.A., Schachte, P., S{\o}ndergaard, H., Stuckey, P.J.: An
  abstract domain of uninterpreted functions. In: Verification, Model Checking,
  and Abstract Interpretation - 17th International Conference, {VMCAI} 2016,
  St. Petersburg, FL, USA, January 17-19, 2016. Proceedings. pp. 85--103
  (2016). \doi{10.1007/978-3-662-49122-5\_4}

\bibitem{DBLP:conf/sas/GangeNSSS16}
Gange, G., Navas, J.A., Schachte, P., S{\o}ndergaard, H., Stuckey, P.J.:
  Exploiting sparsity in difference-bound matrices. In: Rival, X. (ed.) Static
  Analysis - 23rd International Symposium, {SAS} 2016, Edinburgh, UK, September
  8-10, 2016, Proceedings. Lecture Notes in Computer Science, vol.~9837, pp.
  189--211. Springer (2016). \doi{10.1007/978-3-662-53413-7\_10}

\bibitem{gllvm}
gllvm. \url{https://github.com/SRI-CSL/gllvm} (2020)

\bibitem{Granger:IJCM1989}
Granger, P.: Static analysis of arithmetical congruences. International Journal
  of Computer Mathematics  \textbf{30}(3-4),  165--190 (1989).
  \doi{10.1080/00207168908803778}

\bibitem{halbwachsSAS2012}
Halbwachs, N., Henry, J.: When the decreasing sequence fails. In: Static
  Analysis - 19th International Symposium, {SAS} 2012, Deauville, France,
  September 11-13, 2012. Proceedings. pp. 198--213 (2012).
  \doi{10.1007/978-3-642-33125-1\_15}

\bibitem{DBLP:journals/fmsd/HalbwachsMG06}
Halbwachs, N., Merchat, D., Gonnord, L.: Some ways to reduce the space
  dimension in polyhedra computations. Formal Methods Syst. Des.
  \textbf{29}(1),  79--95 (2006). \doi{10.1007/s10703-006-0013-2}

\bibitem{DBLP:journals/entcs/HenryMM12}
Henry, J., Monniaux, D., Moy, M.: {PAGAI:} {A} path sensitive static analyser.
  Electron. Notes Theor. Comput. Sci.  \textbf{289},  15--25 (2012).
  \doi{10.1016/j.entcs.2012.11.003}

\bibitem{DBLP:conf/sas/HeoOY16}
Heo, K., Oh, H., Yang, H.: Learning a variable-clustering strategy for octagon
  from labeled data generated by a static analysis. In: Rival, X. (ed.) Static
  Analysis - 23rd International Symposium, {SAS} 2016, Edinburgh, UK, September
  8-10, 2016, Proceedings. Lecture Notes in Computer Science, vol.~9837, pp.
  237--256. Springer (2016). \doi{10.1007/978-3-662-53413-7\_12}

\bibitem{DBLP:conf/cav/JeannetM09}
Jeannet, B., Min{\'{e}}, A.: Apron: {A} library of numerical abstract domains
  for static analysis. In: Bouajjani, A., Maler, O. (eds.) Computer Aided
  Verification, 21st International Conference, {CAV} 2009, Grenoble, France,
  June 26 - July 2, 2009. Proceedings. Lecture Notes in Computer Science,
  vol.~5643, pp. 661--667. Springer (2009). \doi{10.1007/978-3-642-02658-4\_52}

\bibitem{DBLP:journals/pacmpl/KimVT20}
Kim, S.K., Venet, A.J., Thakur, A.V.: Deterministic parallel fixpoint
  computation. {PACMPL}  \textbf{4}(POPL),  14:1--14:33 (2020).
  \doi{10.1145/3371082}

\bibitem{DBLP:conf/popl/LiAKGC14}
Li, Y., Albarghouthi, A., Kincaid, Z., Gurfinkel, A., Chechik, M.: Symbolic
  optimization with {SMT} solvers. In: Jagannathan, S., Sewell, P. (eds.) The
  41st Annual {ACM} {SIGPLAN-SIGACT} Symposium on Principles of Programming
  Languages, {POPL} '14, San Diego, CA, USA, January 20-21, 2014. pp. 607--618.
  {ACM} (2014). \doi{10.1145/2535838.2535857}

\bibitem{DBLP:journals/ftpl/Mine17}
Min{\'{e}}, A.: Tutorial on static inference of numeric invariants by abstract
  interpretation. Foundations and Trends in Programming Languages
  \textbf{4}(3-4),  120--372 (2017). \doi{10.1561/2500000034}

\bibitem{monniaux2005parallel}
Monniaux, D.: The parallel implementation of the astr{\'{e}}e static analyzer.
  In: Programming Languages and Systems, Third Asian Symposium, {APLAS} 2005,
  Tsukuba, Japan, November 2-5, 2005, Proceedings. pp. 86--96 (2005).
  \doi{10.1007/11575467\_7}

\bibitem{DBLP:conf/cc/NaeemLR10}
Naeem, N.A., Lhot{\'{a}}k, O., Rodriguez, J.: Practical extensions to the
  {IFDS} algorithm. In: Gupta, R. (ed.) Compiler Construction, 19th
  International Conference, {CC} 2010, Held as Part of the Joint European
  Conferences on Theory and Practice of Software, {ETAPS} 2010, Paphos, Cyprus,
  March 20-28, 2010. Proceedings. Lecture Notes in Computer Science, vol.~6011,
  pp. 124--144. Springer (2010). \doi{10.1007/978-3-642-11970-5\_8}

\bibitem{crab}
Navas, J.A.: Crab: Cornucopia of abstractions: a language-agnostic library for
  abstract interpretation. \url{https://github.com/seahorn/crab} (2019)

\bibitem{DBLP:journals/toplas/OhHLLPKY14}
Oh, H., Heo, K., Lee, W., Lee, W., Park, D., Kang, J., Yi, K.: Global sparse
  analysis framework. {ACM} Trans. Program. Lang. Syst.  \textbf{36}(3),
  8:1--8:44 (2014). \doi{10.1145/2590811}

\bibitem{DBLP:conf/pldi/OhHLLY12}
Oh, H., Heo, K., Lee, W., Lee, W., Yi, K.: Design and implementation of sparse
  global analyses for c-like languages. In: {ACM} {SIGPLAN} Conference on
  Programming Language Design and Implementation, {PLDI} '12, Beijing, China -
  June 11 - 16, 2012. pp. 229--238 (2012). \doi{10.1145/2254064.2254092}

\bibitem{okasaki1998fast}
Okasaki, C., Gill, A.: Fast mergeable integer maps. In: Workshop on ML. pp.
  77--86 (1998)

\bibitem{ramalingamTOPLAS1999}
Ramalingam, G.: Identifying loops in almost linear time. {ACM} Trans. Program.
  Lang. Syst.  \textbf{21}(2),  175--188 (1999). \doi{10.1145/316686.316687}

\bibitem{ramalingamTOPLAS2002}
Ramalingam, G.: On loops, dominators, and dominance frontiers. {ACM} Trans.
  Program. Lang. Syst.  \textbf{24}(5),  455--490 (2002).
  \doi{10.1145/570886.570887}

\bibitem{rastello:hal-00761555}
Rastello, F.: {On Sparse Intermediate Representations: Some Structural
  Properties and Applications to Just-In-Time Compilation}. University works,
  {Inria Grenoble Rh{\^o}ne-Alpes} (Dec 2012),
  \url{https://hal.inria.fr/hal-00761555}, habilitation {\`a} diriger des
  recherches, {\'E}cole normale sup{\'e}rieure de Lyon

\bibitem{RHSPOPL1995}
Reps, T.W., Horwitz, S., Sagiv, M.: Precise interprocedural dataflow analysis
  via graph reachability. In: Conference Record of POPL'95: 22nd {ACM}
  {SIGPLAN-SIGACT} Symposium on Principles of Programming Languages, San
  Francisco, California, USA, January 23-25, 1995. pp. 49--61 (1995).
  \doi{10.1145/199448.199462}

\bibitem{DBLP:conf/vmcai/RepsSY04}
Reps, T.W., Sagiv, S., Yorsh, G.: Symbolic implementation of the best
  transformer. In: Steffen, B., Levi, G. (eds.) Verification, Model Checking,
  and Abstract Interpretation, 5th International Conference, {VMCAI} 2004,
  Venice, Italy, January 11-13, 2004, Proceedings. Lecture Notes in Computer
  Science, vol.~2937, pp. 252--266. Springer (2004).
  \doi{10.1007/978-3-540-24622-0\_21}

\bibitem{DBLP:conf/vmcai/RepsT16}
Reps, T.W., Thakur, A.V.: Automating abstract interpretation. In: Jobstmann,
  B., Leino, K.R.M. (eds.) Verification, Model Checking, and Abstract
  Interpretation - 17th International Conference, {VMCAI} 2016, St. Petersburg,
  FL, USA, January 17-19, 2016. Proceedings. Lecture Notes in Computer Science,
  vol.~9583, pp. 3--40. Springer (2016). \doi{10.1007/978-3-662-49122-5\_1}

\bibitem{DBLP:conf/ppdp/SeidlV18}
Seidl, H., Vogler, R.: Three improvements to the top-down solver. In: Sabel,
  D., Thiemann, P. (eds.) Proceedings of the 20th International Symposium on
  Principles and Practice of Declarative Programming, {PPDP} 2018, Frankfurt am
  Main, Germany, September 03-05, 2018. pp. 21:1--21:14. {ACM} (2018).
  \doi{10.1145/3236950.3236967}

\bibitem{DBLP:conf/pldi/SinghPV15}
Singh, G., P{\"{u}}schel, M., Vechev, M.T.: Making numerical program analysis
  fast. In: Proceedings of the 36th {ACM} {SIGPLAN} Conference on Programming
  Language Design and Implementation, Portland, OR, USA, June 15-17, 2015. pp.
  303--313 (2015). \doi{10.1145/2737924.2738000}

\bibitem{DBLP:conf/popl/SinghPV17}
Singh, G., P{\"{u}}schel, M., Vechev, M.T.: Fast polyhedra abstract domain. In:
  Castagna, G., Gordon, A.D. (eds.) Proceedings of the 44th {ACM} {SIGPLAN}
  Symposium on Principles of Programming Languages, {POPL} 2017, Paris, France,
  January 18-20, 2017. pp. 46--59. {ACM} (2017). \doi{10.1145/3009837.3009885}

\bibitem{DBLP:conf/cav/SinghPV18}
Singh, G., P{\"{u}}schel, M., Vechev, M.T.: Fast numerical program analysis
  with reinforcement learning. In: Computer Aided Verification - 30th
  International Conference, {CAV} 2018, Held as Part of the Federated Logic
  Conference, FloC 2018, Oxford, UK, July 14-17, 2018, Proceedings, Part {I}.
  pp. 211--229 (2018). \doi{10.1007/978-3-319-96145-3\_12}

\bibitem{DBLP:journals/pacmpl/SinghPV18}
Singh, G., P{\"{u}}schel, M., Vechev, M.T.: A practical construction for
  decomposing numerical abstract domains. Proc. {ACM} Program. Lang.
  \textbf{2}({POPL}),  55:1--55:28 (2018). \doi{10.1145/3158143}

\bibitem{DBLP:journals/jacm/Tarjan79}
Tarjan, R.E.: Applications of path compression on balanced trees. J. {ACM}
  \textbf{26}(4),  690--715 (1979). \doi{10.1145/322154.322161}

\bibitem{codehawk}
Technology, K.: Codehawk. \url{https://github.com/kestreltechnology/codehawk}
  (2020)

\bibitem{DBLP:conf/sas/ThakurER12}
Thakur, A.V., Elder, M., Reps, T.W.: Bilateral algorithms for symbolic
  abstraction. In: Min{\'{e}}, A., Schmidt, D. (eds.) Static Analysis - 19th
  International Symposium, {SAS} 2012, Deauville, France, September 11-13,
  2012. Proceedings. Lecture Notes in Computer Science, vol.~7460, pp.
  111--128. Springer (2012). \doi{10.1007/978-3-642-33125-1\_10}

\bibitem{DBLP:journals/entcs/ThakurLLR15}
Thakur, A.V., Lal, A., Lim, J., Reps, T.W.: Posthat and all that: Automating
  abstract interpretation. Electron. Notes Theor. Comput. Sci.  \textbf{311},
  15--32 (2015). \doi{10.1016/j.entcs.2015.02.003}

\bibitem{DBLP:conf/cav/ThakurR12}
Thakur, A.V., Reps, T.W.: A method for symbolic computation of abstract
  operations. In: Madhusudan, P., Seshia, S.A. (eds.) Computer Aided
  Verification - 24th International Conference, {CAV} 2012, Berkeley, CA, USA,
  July 7-13, 2012 Proceedings. Lecture Notes in Computer Science, vol.~7358,
  pp. 174--192. Springer (2012). \doi{10.1007/978-3-642-31424-7\_17}

\bibitem{PLDI:VB2004}
Venet, A., Brat, G.P.: Precise and efficient static array bound checking for
  large embedded {C} programs. In: Proceedings of the {ACM} {SIGPLAN} 2004
  Conference on Programming Language Design and Implementation 2004,
  Washington, DC, USA, June 9-11, 2004. pp. 231--242 (2004).
  \doi{10.1145/996841.996869}

\bibitem{wangASPLOS2017}
Wang, K., Hussain, A., Zuo, Z., Xu, G.H., Sani, A.A.: Graspan: {A}
  single-machine disk-based graph system for interprocedural static analyses of
  large-scale systems code. In: Proceedings of the Twenty-Second International
  Conference on Architectural Support for Programming Languages and Operating
  Systems, {ASPLOS} 2017, Xi'an, China, April 8-12, 2017. pp. 389--404 (2017).
  \doi{10.1145/3037697.3037744}

\bibitem{WeissRL:ICSE2015}
Weiss, C., Rubio{-}Gonz{\'{a}}lez, C., Liblit, B.: Database-backed program
  analysis for scalable error propagation. In: 37th {IEEE/ACM} International
  Conference on Software Engineering, {ICSE} 2015, Florence, Italy, May 16-24,
  2015, Volume 1. pp. 586--597 (2015). \doi{10.1109/ICSE.2015.75}

\bibitem{DBLP:conf/ipps/0002GJWHWL19}
Zuo, Z., Gu, R., Jiang, X., Wang, Z., Huang, Y., Wang, L., Li, X.: Bigspa: An
  efficient interprocedural static analysis engine in the cloud. In: 2019
  {IEEE} International Parallel and Distributed Processing Symposium, {IPDPS}
  2019, Rio de Janeiro, Brazil, May 20-24, 2019. pp. 771--780. {IEEE} (2019).
  \doi{10.1109/IPDPS.2019.00086}

\end{thebibliography}
\clearpage
\appendix

\section{Proofs}
\label{app:appendix-proof}
This section provides proofs of theorems presented in the paper.

\subsection{Nesting forest $(V, \nesting)$ and total order $(V, \luo)$ in \pref{sec:Algorithm}}
\label{app:appendix-order}
This section presents the theorems and proofs about $\nesting$ and $\luo$
defined in \pref{sec:Algorithm}.

A partial order $(S, \relation)$ is a forest if for all $x \in S$,
$(\postset{x}_{\relation}, R)$ is a chain, where $\postset{x}_{\relation} \eqdef
\{ y\in S \mid x\ \relation\ y\}$.

\begin{theorem}
  \label{thm:forest}
  $(V, \nesting)$ is a forest.
\end{theorem}
\begin{proof}
  First, we show that $(V, \nesting)$ is a partial order. Let $x, y, z$ be a
  vertex in $V$.
  \begin{itemize}
    \item Reflexivity: $x \nesting x$. This is true by the definition of
    $\nesting$.
    \item Transitivity: $x \nesting y$ and $y \nesting z$ implies $x \nesting
    z$. (i)~If $x = y$, $x \nesting z$. (ii)~Otherwise, by definition of
    $\nesting$, $y \in \omega(x)$. Furthermore, (ii-1) if $y = z$, $z \in
    \omega(x)$; and hence, $x \nesting z$. (ii-2) Otherwise, $z \in \omega(y)$,
    and by definition of HTO, $z \in \omega(x)$.
    \item Anti-symmetry: $x \nesting y$ and $y \nesting x$ implies $x = y$.
    Suppose $x \neq y$. By definition of $\nesting$ and premises, $y \in
    \omega(x)$ and $x \in \omega(y)$. Then, by definition of HTO, $x \prec y$
    and $y \prec x$. This contradicts that $\preceq$ is a total order.
  \end{itemize}

  Next, we show that the partial order is a forest. Suppose there exists $v \in
  V$ such that $(\postset{v}_{\nesting}, \nesting)$ is not a chain. That is,
  there exists $x, y \in \postset{v}_{\nesting}$ such that $x \not\nesting y$
  and $y \not\nesting x$. Then, by definition of HTO, $\components(x) \cap
  \components(y) = \emptyset$. However, this contradicts that $v \in
  \components(x)$ and $v \in \components(y)$.
  \qed
\end{proof}

\begin{theorem}
  \label{thm:total}
  $(V, \luo)$ is a total order.
\end{theorem}
\begin{proof}
  We prove the properties of a total order. Let $x, y, z$ be a vertex in $V$.
  \begin{itemize}
    \item Connexity: $x \luo y$ or $y \luo x$. This follows from the connexity
    of the total order~$\preceq$.
    \item Transitivity: $x \luo y$ and $y \luo z$ implies $x \luo z$.
    (i)~Suppose $x \nesting y$. (i-1)~If $y \nesting z$, by transitivity of
    $\nesting$, $x \nesting z$. (ii-2)~Otherwise, $z \not\nesting y$ and $y
    \preceq z$. It cannot be $z \nesting x$ because transitivity of $\nesting$
    implies $z \nesting y$, which is a contradiction. Furthermore, it cannot be
    $z \prec x$ because $y \preceq z \prec x$ and $x \nesting y$ implies $y \in
    \omega(z)$ by the definition of HTO. By connexity of $\preceq$, $x \preceq
    z$. (ii)~Otherwise $y \not\nesting x$ and $x \preceq y$. (ii-1)~If $y
    \nesting z$, $z \not\nesting x$ because, otherwise, transitivity of
    $\nesting$ will imply $y \nesting x$. By connexity of $\preceq$, it is
    either $x \preceq z$ or $z \prec x$. If $x \preceq z$, $x \luo z$. If $z
    \prec x$, by definition of HTO, $z \in \omega(z)$.
    \item Anti-symmetry: $x \luo y$ and $y \luo x$ implies $x = y$. (i)~If $x
    \nesting y$, it should be $y \nesting x$ for $y \luo x$ to be true. By
    anti-symmetry of $\nesting$, $x = y$. (ii)~Otherwise, $y \not\nesting x$ and
    $x \preceq y$. For $y \luo x$ to be true, $x \not\nesting y$ and $x \preceq
    y$. By anti-symmetry of $\preceq$, $x = y$.
  \end{itemize}
  \qed
\end{proof}

\newpage
\begin{restatable}{theorem}{ThmRead}
  \label{thm:luo}
  For $u, v \in V$, if $\instr[v]$ reads $\post[u]$, then $u \luo v$.
\end{restatable}
\begin{proof}
  By the definition of the mapping $\instr$, there must exists $v' \in V$ such
  that $u \cfgarrow v'$ and $v' \nesting v$ for $\instr[v]$ to read $\post[u]$.
  By the definition of WTO, it is either $u \prec v'$ and $v' \notin \omega(u)$,
  or $v' \preceq u$ and $v' \in \omega(u)$. In both cases, $u \luo v'$. Because
  $v' \nesting v$, and hence $v' \luo v$, $u \luo v$.
  \qed
\end{proof}

\subsection{Optimality of $\memconfigopt$ in \pref{sec:Algorithm}}
\label{app:appendix-memopt}
This section presents the theorems and proofs about the optimality of
$\memconfigopt$ described in \pref{sec:Algorithm}. The theorem is divided into
optimality theorems of the maps that constitute $\memconfigopt$.

Given $\memconfig(\dpost, \achk, \dpostl, \dprel)$ and a map $\dpost_0$, we use
$\memconfig \lightning \dpost_0$ to denote the memory configuration ($\dpost_0$,
$\achk$, $\dpostl$, $\dprel$). Similarly, $\memconfig \lightning \achk_0$ means
($\dpost$, $\achk_0$, $\dpostl$, $\dprel$), and so on. For a given $\FM$ program
$P$, each map $X$ that constitutes a memory configuration is valid for $P$ iff
$\memconfig \lightning X$ is valid for every valid memory
configuration~$\memconfig$. Also, $X$ is optimal for $P$ iff $\memconfig
\lightning X$ is optimal for an optimal memory configuration~$\memconfig$.

\begin{restatable}{theorem}{ThmDpostValid}
  \label{thm:dpostvalid}
  $\dpostopt$ is valid. That is, given an $\FM$ program $P$ and a valid memory
  configuration $\memconfig$, $\llbracket P \rrbracket_{\memconfig \lightning
  \dpostopt} = \llbracket P \rrbracket_{\memconfig}$.
\end{restatable}
\begin{proof}
  Our approach does not change the iteration order and only changes where the
  deallocations are performed. Therefore, it is sufficient to show that for all
  $u \cfgarrow v$, $\post[u]$ is available whenever $\instr[v]$ is executed.

  Suppose that this is false: there exists an edge $u \cfgarrow v$ that violates
  it. Let $d$ be $\dpostopt[u]$ computed by our approach. Then, the execution
  trace of $P$ has execution of $\instr[v]$ after the deallocation of $\post[u]$
  in $\instr[d]$, with no execution of $\instr[u]$ in between.

  Because $\luo$ is a total order, it is either $d \luoneq v$ or $v \luo d$. It
  must be $v \luo d$, because $d \luoneq v$ implies $d \luoneq v \luo \foo(u,
  v)$, which contradicts the definition of $\dpostopt[u]$. Then, by definition
  of $\luo$, it is either $v \nesting d$ or $(d \not\nesting v) \wedge (v
  \preceq d)$. In both cases, the only way $\instr[v]$ can be executed after
  $\instr[d]$ is to have another head $h$ whose $\fmrepeat$ instruction includes
  both $\instr[d]$ and $\instr[v]$. That is, when $d \nestingneq h$ and $v
  \nestingneq h$.
  By definition of WTO and $u \cfgarrow v$, it is either $u \prec v$, or $u
  \nesting v$. It must be $u \prec v$, because if $u \nesting v$, $\instr[u]$ is
  part of $\instr[v]$, making $\instr[u]$ to be executed before reading
  $\post[u]$ in $\instr[v]$. Furthermore, it must be $u \prec h$, because if $h
  \preceq u$, $\instr[u]$ is executed before $\instr[v]$ in each iteration over
  $\components(h)$. However, that implies $h \in (\postset{v}_{\nesting}
  \setminus \postset{u}_{\nesting})$, which combined with $d \nestingneq h$,
  contradicts the definition of $\dpostopt[u]$. Therefore, no such edge $u
  \cfgarrow v$ can exist and the theorem is true.
  \qed
\end{proof}

\begin{restatable}{theorem}{ThmDpostOpt}
  \label{thm:dpostopt}
  $\dpostopt$ is optimal. That is, given an $\FM$ program $P$, memory footprint
  of $\llbracket P \rrbracket_{\memconfig \lightning \dpostopt}$ is smaller than
  or equal to that of $\llbracket P \rrbracket_{\memconfig}$ for all valid
  memory configuration $\memconfig$.
\end{restatable}
\begin{proof}
  For $\dpostopt$ to be optimal, deallocation of $\post$ values must be
  determined at earliest positions as possible with a valid memory configuration
  $\memconfig \lightning \dpostopt$. That is, there should not exists $u, b \in
  V$ such that if $d = \dpostopt[u]$, $b \neq d$, $\memconfig \lightning
  (\dpostopt[u \leftarrow b])$ is valid, and $\instr[b]$ deletes $\post[u]$
  earlier than $\instr[d]$.

  Suppose that this is false: such $u, b$ exists. Let $d$ be $\dpostopt[u]$,
  computed by our approach. Then it must be $b \luoneq d$ for $\instr[b]$ to be
  able to delete $\post[u]$ earlier than $\instr[d]$. Also, for all $u \cfgarrow
  v$, it must be $v \luo b$ for $\instr[v]$ to be executed before deleting
  $\post[u]$ in $\instr[b]$.

  By definition of $\dpostopt$, $v \luo d$ for all $u \cfgarrow v$. Also, by
  \pref{thm:luo}, $u \luo v$. Hence, $u \luo d$, making it either $u \nesting
  d$, or $(d \not\nesting u) \wedge (u \preceq d)$. If $u \nesting d$, by
  definition of $\foo$, it must be $u \cfgarrow d$. Therefore, it must be $d
  \luo b$, which contradicts that $b \luoneq d$. Alternative, if $(d
  \not\nesting u) \wedge (u \preceq d)$, there must exist $v \in V$ such that $u
  \cfgarrow v$ and $\foo(u, v) = d$. To satisfy $v \luo b$, $v \nesting d$, and
  $b \luoneq d$, it must be $b \nesting d$. However, this makes the analysis
  incorrect because when stabilization check fails for $\components(d)$,
  $\instr[v]$ gets executed again, attempting to read $\post[u]$ that is already
  deleted by $\instr[b]$. Therefore, no such $u, b$ can exist, and the theorem
  is true.
  \qed
\end{proof}

\begin{restatable}{theorem}{ThmAchkValid}
  \label{thm:achkvalid}
  $\achkopt$ is valid. That is, given an $\FM$ program $P$ and a valid memory
  configuration $\memconfig$, $\llbracket P \rrbracket_{\memconfig \lightning
  \achkopt} = \llbracket P \rrbracket_{\memconfig}$
\end{restatable}
\begin{proof}
  Let $v = \achkopt[u]$. If $v$ is a head, by definition of $\achkopt$,
  $\components(v)$ is the largest component that contains $u$. Therefore, once
  $\components(v)$ is stabilized, $\instr[u]$ can no longer be executed, and
  $\pre[u]$ remains the same. If $v$ is not a head, then $v = u$. That is, there
  is no component that contains $u$. Therefore, $\pre[u]$ remains the same after
  the execution of $\instr[u]$. In both cases, the value passed to $\checkmap_u$
  are the same as when using $\achkdef$.
  \qed
\end{proof}

\begin{restatable}{theorem}{ThmAchkOpt}
  \label{thm:achkopt}
  $\achkopt$ is optimal. That is, given an $\FM$ program $P$, memory footprint of
  $\llbracket P \rrbracket_{\memconfig \lightning \achkopt}$ is smaller than or
  equal to that of $\llbracket P \rrbracket_{\memconfig}$ for all valid memory
  configuration $\memconfig$.
\end{restatable}
\begin{proof}
  Because $\pre$ value is deleted right after its corresponding assertions are
  checked, it is sufficient to show that assertion checks are placed at the
  earliest positions with $\achkopt$.

  Let $v = \achkopt[u]$. By definition of $\achkopt$, $u \nesting v$. For some
  $b$ to perform assertion checks of $u$ earlier than $v$, it must satisfy $b
  \nestingneq v$.
  However, because one cannot know in advance when a component of $v$ would
  stabilize and when $\pre[u]$ would converge, the assertion checks of $u$
  cannot be performed in $\instr[b]$.
  Therefore, our approach puts the assertion checks at the earliest positions,
  and it leads to the minimum memory footprint.
  \qed
\end{proof}

\begin{restatable}{theorem}{ThmDpostlValid}
  \label{thm:dpostlvalid}
  $\dpostlopt$ is valid. That is, given an $\FM$ program $P$ and a valid memory
  configuration $\memconfig$, $\llbracket P \rrbracket_{\memconfig \lightning
  \dpostlopt} = \llbracket P \rrbracket_{\memconfig}$.
\end{restatable}
\begin{proof}
  Again, our approach does not change the iteration order and only changes where
  the deallocations are performed. Therefore, it is sufficient to show that for
  all $u \cfgarrow v$, $\post[u]$ is available whenever $\instr[v]$ is executed.

  Suppose that this is false: there exists an edge $u \cfgarrow v$ that violates
  it. Let $d'$ be element in $\dpostlopt[u]$ that causes this violation. Then,
  the execution trace of $P$ has execution of $\instr[v]$ after the deallocation
  of $\post[u]$ in $\instr[d']$, with no execution of $\instr[u]$ in between.
  Because $\post[u]$ is deleted inside the loop of $\instr[d']$, $\instr[v]$
  must be nested in $\instr[d']$ or be executed after $\instr[d']$ to be
  affected. That is, it must be either $v \nesting d'$ or $d' \prec v$. Also,
  because of how $\dpostlopt[u]$ is computed, $u \nesting d'$.

  First consider the case $v \nesting d'$. By definition of WTO and $u \cfgarrow
  v$, it is either $u \prec v$ or $u \nesting v$. In either case, $\instr[u]$
  gets executed before $\instr[v]$ reads $\post[u]$. Therefore, deallocation of
  $\post[u]$ in $\instr[d']$ cannot cause the violation.

  Alternatively, consider $d' \prec v$ and $v \not\nesting d'$. Because $u
  \nesting d'$, $\post[u]$ is generated in each iteration over
  $\components(d')$, and the last iteration does not delete $\post[u]$.
  Therefore, $\post[u]$ will be available when executing $\instr[v]$. Therefore,
  such $u, d'$ does not exists, and the theorem is true.
  \qed
\end{proof}

\begin{restatable}{theorem}{ThmDpostlOpt}
  \label{thm:dpostlopt}
  $\dpostlopt$ is optimal. That is, given an $\FM$ program $P$, memory footprint
  of $\llbracket P \rrbracket_{\memconfig \lightning \dpostlopt}$ is smaller
  than or equal to that of $\llbracket P \rrbracket_{\memconfig}$ for all valid
  memory configuration $\memconfig$.
\end{restatable}
\begin{proof}
  Because one cannot know when a component would stabilize in advance, the
  decision to delete intermediate $\post[u]$ cannot be made earlier than the
  stabilization check of a component that contains $u$. Our approach makes such
  decisions in all relevant components that contains $u$.

  If $u \nesting d$, $\dpostlopt[u] = \postset{u}_{\nesting} \cap
  \preset{d}_{\nesting}$. Because $\post[u]$ is deleted in $\instr[d]$, we do
  not have to consider components in $\postset{d}_{\nesting} \setminus \{d\}$.
  Alternatively, if $u \not\nesting d$, $\dpostlopt[u] = \postset{u}_{\nesting}
  \setminus \postset{d}_{\nesting}$. Because $\post[u]$ is deleted $\instr[d]$,
  we do not have to consider components in $\postset{u}_{\nesting} \setminus
  \postset{d}_{\nesting}$. Therefore, $\dpostlopt$ is optimal.
  \qed
\end{proof}

\begin{restatable}{theorem}{ThmDprelValid}
  \label{thm:dprelvalid}
  $\dprelopt$ is valid. That is, given an $\FM$ program $P$ and a valid memory
  configuration $\memconfig$, $\llbracket P \rrbracket_{\memconfig \lightning
  \dprelopt} = \llbracket P \rrbracket_{\memconfig}$.
\end{restatable}
\begin{proof}
  $\pre[u]$ is only used in assertion checks and to
  perform widening in $\instr[u]$. Because $u$ is removed from $\dprel[u]$, the
  deletion does not affect widening.

  For all $v \in \dprel[u]$, $v \nesting \achkopt[u]$. Because $\pre[u]$ is not
  deleted when $\components(v)$ is stabilized, $\pre[u]$ will be available when
  performing assertion checks in $\instr[\achkopt[u]]$. Therefore, $\dprel$ is
  valid.
  \qed
\end{proof}

\begin{restatable}{theorem}{ThmDprelOpt}
  \label{thm:dprelopt}
  $\dprelopt$ is optimal. That is, given an $\FM$ program $P$, memory footprint
  of $\llbracket P \rrbracket_{\memconfig \lightning \dprelopt}$ is smaller than
  or equal to that of $\llbracket P \rrbracket_{\memconfig}$ for all valid
  memory configuration $\memconfig$.
\end{restatable}
\begin{proof}
  Because one cannot know when a component would stabilize in advance, the
  decision to delete intermediate $\pre[u]$ cannot be made earlier than the
  stabilization check of a component that contains $u$. Our approach makes such
  decisions in all components that contains $u$. Therefore, $\dprelopt$ is
  optimal.
  \qed
\end{proof}

\ThmMemopt*
\begin{proof}
  This follows from theorems \pref{thm:dpostlvalid} to \ref{thm:dprelopt}.
  \qed
\end{proof}

\subsection{Correctness and efficiency of \algwto{} in \pref{sec:Efficient}}
\label{app:appendix-alg}
This section presents the theorems and proofs about the correctness and
efficiency of \algwto{} (\pref{alg:bu}, \pref{sec:Efficient}).

\ThmAlgCorrect*
\begin{proof} We show that each map is constructed correctly.
  \begin{itemize}
    \item $\dpostopt$: Let $v'$ be the value of $\dpostopt[u]$ before
    overwritten in \pref{li:wto-prec1}, \ref{li:wto-prec2}, or
    \ref{li:wto-nest}. Descending post DFN ordering corresponds to a topological
    sorting of the nested SCCs. Therefore, in \pref{li:wto-prec1} and
    \ref{li:wto-prec2}, $v' \prec v$. Also, because $v \nesting \etch$ for all
    $v \in N_h$ in \pref{li:wto-nest}, $v' \nesting v$. In any case, $v' \luo
    v$. Because $\repKw(v)$ essentially performs $\foo(u,v)$ when restoring the
    edges, the final $\dpostopt[u]$ is the maximum of the lifted successors, and
    the map is correctly computed.
    \item $\dpostlopt$: The correctness follows from the correctness of $\TKw$.
    Because the components are constructed bottom-up, $\repKw(u)$ in
    \pref{li:wto-Tprec1} and \ref{li:wto-Tprec2} returns
    $\text{max}_{\nesting}(\postset{u}_{\nesting} \setminus
    \postset{\dpostopt[u]}_{\nesting})$. Also, $\mathsf{N}^* = \nesting$. Thus,
    $\dpostlopt$ is correctly computed.
    \item $\achkopt$: At the end of the algorithm $\repKw(v)$ is the head of
    maximal component that contains $v$, or $v$ itself when $v$ is outside of
    any components. Therefore, $\achkopt$ is correctly computed.
    \item $\dprelopt$: Using the same reasoning as in $\achkopt$, and because
    $\mathsf{N}^* = \nesting$, $\dprelopt$ is correctly computed.
  \end{itemize}
  \qed
\end{proof}

\ThmAlgFast*
\begin{proof}
  The base WTO-construction algorithm is almost-linear
  time~\cite{DBLP:journals/pacmpl/KimVT20}. The starred lines in \pref{alg:bu}
  visit each edge and vertex once. Therefore, time complexity still remains
  almost-linear time.
  \qed
\end{proof}

\ifarxiv
\newpage
\section{Further experimental evaluation}
\label{app:appendix-eval}

\begin{table} \caption{Measurements for benchmarks that took less than 5 seconds
  are summarized in the table below. Time diff shows the runtime of $\ikos$
  minus that of $\mikos$ (positive means speedup in $\mikos$). Mem diff shows
  the memory footprint of $\ikos$ minus that of $\mikos$ (positive means memory
  reduction in $\mikos$).}
  \let\center\empty
  \let\endcenter\relax
  \centering
  \resizebox{\width}{!}{\begin{center}
\begin{tabular}{@{}lccc|ccc|ccc|ccc@{}}
  \toprule
  $<\!5$s & \multicolumn{3}{c}{Time (s)} & \multicolumn{3}{c}{Memory (MB)} & \multicolumn{3}{c}{Time diff (s)} & \multicolumn{3}{c}{Memory diff (MB)} \\
  & min. & max. & avg. & min. & max. & avg. & min. & max. & avg. & min. & max. & avg. \\
  \midrule
  T1\; & \svcshorttimemin{} & \svcshorttimemax{} & \svcshorttimeavg{} & \svcshortmemmin{} & \svcshortmemmax{} & \svcshortmemavg{} & \svcshorttimem{} & \svcshorttimep{} & \svcshorttimea{} & \svcshortmemm{} & \svcshortmemp{} & \svcshortmema{} \\
  T2\; & \ossshorttimemin{} & \ossshorttimemax{} & \ossshorttimeavg{} & \ossshortmemmin{} & \ossshortmemmax{} & \ossshortmemavg{} & \ossshorttimem{} & \ossshorttimep{} & \ossshorttimea{} & \ossshortmemm{} & \ossshortmemp{} & \ossshortmema{} \\
  \bottomrule
\end{tabular}
\end{center}
}
\end{table}

\begin{table} \caption{\textbf{Task T1.} A sample of the results for task T1 in
  \pref{fig:mscat-svc}, excluding the non-completed benchmarks in $\ikos$. The
  first 5 rows list benchmarks with the smallest memory reduction ratio (MRR)s.
  The latter 5 rows list benchmarks with the largest memory footprints. The
  smaller the MRR, the greater the reduction in memory footprint. T: time; MF:
  memory footprint.}
  \label{tab:rq1-svc}
  \let\center\empty
  \let\endcenter\relax
  \centering
  \resizebox{\width}{!}{\begin{tabular}{l|rr|rr|r}
\toprule
     &                                      \multicolumn{2}{c}{$\ikos$} & \multicolumn{2}{c}{$\mikos$} &   \\
                                         Benchmark &  T (s) &  MF (MB) &  T (s) &  MF (MB) &  MRR \\
\midrule
3.16-rc1/205\_9a-net-rtl8187  &         1500 &   45905 &         1314 &          56 &      0.001 \\
4.2-rc1/43\_2a-mmc-rtsx       &        786.5 &   26909 &        594.8 &          42 &      0.002 \\
4.2-rc1/43\_2a-video-radeonfb &         2494 &   56752 &         1930 &         107 &      0.002 \\
4.2-rc1/43\_2a-net-skge       &         3523 &   47392 &         3131 &          98 &      0.002 \\
4.2-rc1/43\_2a-usb-hcd        &        220.4 &   17835 &        150.8 &          39 &      0.002 \\
\midrule
\midrule
4.2-rc1/32\_7a-target\_core\_mod &         1316 &   60417 &         1110 &        2967 &      0.049 \\
challenges/3.14-alloc-libertas   &         2094 &   60398 &         1620 &         626 &      0.010 \\
4.2-rc1/43\_2a-net-libertas      &         1634 &   59902 &         1307 &         307 &      0.005 \\
challenges/3.14-kernel-libertas  &         2059 &   59826 &         1688 &        2713 &      0.045 \\
3.16-rc1/43\_2a-sound-cs46xx     &         3101 &   58087 &         2498 &         193 &      0.003 \\
\bottomrule
\end{tabular}
}
\end{table}
\begin{table} \caption{\textbf{Task T1.} A sample of the results for task T1 in
  \pref{fig:tscat-svc}. The first 3 rows list benchmarks with lowest speedups.
  The latter 3 rows list benchmarks with highest speedups. T: time; MF: memory
  footprint.}
  \label{tab:rq2-svc}
  \let\center\empty
  \let\endcenter\relax
  \centering
  \resizebox{\width}{!}{\begin{tabular}{l|rr|rr|rr}
\toprule
     &                                      \multicolumn{2}{c}{$\ikos$} & \multicolumn{2}{c}{$\mikos$} &   &   \\
                                         Benchmark & T (s) &  MF (MB) &  T (s) &  MF (MB) &  MRR &   Speedup \\
\midrule
  challenges/3.8-usb-main11   &           42.63 &        541 &           48.92 &         122 &      0.225 &  0.87$\times$ \\
  challenges/3.8-usb-main0    &           54.31 &       3025 &           61.78 &         190 &      0.063 &  0.88$\times$ \\
  challenges/3.8-usb-main1    &           42.84 &        457 &           47.73 &         119 &      0.261 &  0.90$\times$ \\
\midrule
\midrule
  3.14/complex-kernel-tm6000  &          745.4 &   25903 &          413.4 &         234 &      0.009 &  1.80$\times$ \\
  4.2-rc1/43\_2a-scsi-st      &          214.6 &   20817 &          119.6 &         547 &      0.026 &  1.79$\times$ \\
  3.14/kernel--rtl8723ae      &          111.9 &     154 &          62.48 &         115 &      0.746 &  1.79$\times$ \\
\bottomrule
\end{tabular}
}
\end{table}

\begin{table} \caption{\textbf{Task T2.} A sample of the results for task T2 in
  \pref{fig:mscat-oss}, excluding the non-completed benchmarks in $\ikos$. The
  first 5 rows list benchmarks with the smallest memory reduction ratio (MRR)s.
  The latter 5 rows list benchmarks with the largest memory footprints. The
  smaller the MRR, the greater the reduction in memory footprint. T: time; MF:
  memory footprint.}
  \label{tab:rq1-oss}
  \let\center\empty
  \let\endcenter\relax
  \centering
  \resizebox{\width}{!}{\begin{tabular}{l|rr|rr|r}
\toprule
     &                                      \multicolumn{2}{c}{$\ikos$} & \multicolumn{2}{c}{$\mikos$} &   \\
                                         Benchmark &  T (s) &  MF (MB) &  T (s) &  MF (MB) &  MRR \\
\midrule
lxsession-0.5.4/lxsession &          146.1 &    5831 &          81.57 &         130 &      0.022 \\
rox-2.11/ROX-Filer        &          362.3 &    9569 &          400.6 &         329 &      0.034 \\
tor-0.3.5.8/tor-resolve   &          58.36 &    1930 &          53.10 &          70 &      0.036 \\
openssh-8.0p1/ssh-keygen  &           1212 &   29670 &           1170 &        1128 &      0.038 \\
xsane-0.999/xsane         &          499.8 &   10118 &          467.5 &         430 &      0.042 \\
\midrule
\midrule
openssh-8.0p1/sftp       &        3036 &   45903 &         3446 &     9137 &      0.199 \\
metacity-3.30.1/metacity &        2111 &   36324 &         2363 &     6329 &      0.174 \\
links-2.19/links         &        2512 &   29761 &         2740 &     3930 &      0.132 \\
openssh-8.0p1/ssh-keygen &        1212 &   29670 &         1170 &     1128 &      0.038 \\
links-2.19/xlinks        &        2523 &   29587 &         2760 &     3921 &      0.133 \\
\bottomrule
\end{tabular}}
\end{table}
\begin{table} \caption{\textbf{Task T2.} A sample of the results for task T2 in
  \pref{fig:tscat-oss}. The first 3 rows list benchmarks with lowest speedups.
  The latter 3 rows list benchmarks with highest speedups. T: time; MF: memory
  footprint.}
  \label{tab:rq2-oss}
  \let\center\empty
  \let\endcenter\relax
  \centering
  \resizebox{\width}{!}{\begin{tabular}{l|rr|rr|rr}
\toprule
     &                                      \multicolumn{2}{c}{$\ikos$} & \multicolumn{2}{c}{$\mikos$} &   &   \\
                                         Benchmark & T (s) &  MF (MB) &  T (s) &  MF (MB) &  MRR &   Speedup \\
\midrule
  moserial-3.0.12/moserial        &        422.3 &        109 &        585.5 &         107 &      0.980 &  0.72$\times$ \\
  openssh-8.0p1/ssh-pkcs11-helper &        82.70 &        674 &        94.61 &         613 &      0.910 &  0.87$\times$ \\
  openssh-8.0p1/sftp              &         3036 &      45903 &         3446 &        9137 &      0.199 &  0.88$\times$ \\
\midrule
\midrule
  packeth-1.9/packETH             &          188.7 &        153 &           83.82 &         120 &      0.782 &  2.25$\times$ \\
  lxsession-0.5.4/lxsession       &          146.1 &       5831 &           81.57 &         130 &      0.022 &  1.79$\times$ \\
  xscreensaver-5.42/braid         &           6.48 &        203 &            4.87 &          36 &      0.179 &  1.33$\times$ \\
\bottomrule
\end{tabular}
}
\end{table}

\fi
\end{document}